\documentclass[a4paper, onecolumn, accepted=2024-06-28]{quantumarticle}
\pdfoutput=1
\usepackage{caption}

\usepackage{subcaption}
\usepackage[utf8]{inputenc}

\usepackage{amssymb}
\usepackage{amsfonts}
\usepackage{amsmath}
\usepackage{amsthm}
\usepackage{mathtools}

\usepackage{graphicx}
\usepackage{indentfirst}
\usepackage{enumerate}
\usepackage{url}

\usepackage{newpxtext}

\usepackage[symbol]{footmisc}

\usepackage{caption}
\usepackage{subcaption}
\usepackage{pgfplots}

\usepackage{tikz}
\definecolor{darkblue}{rgb}{0.15,0.35,0.55}
\definecolor{reddish}{rgb}{.8, 0.2, 0.2}
\definecolor{plotblue}{RGB}{0,119,187}
\definecolor{plotgreen}{RGB}{0,153,136}
\definecolor{plotorange}{RGB}{238,119,51}
\definecolor{plotmagenta}{RGB}{238,51,119}
\definecolor{plotgray}{RGB}{128,128,128}
\definecolor{plotcyan}{RGB}{51,187,238}
\definecolor{plotred}{RGB}{204,51,17}
\usepgfplotslibrary{fillbetween}

\usepackage[backend=bibtex, style = numeric, doi = false, url = true, isbn =
false, maxbibnames = 6, sorting=none]{biblatex}
\bibliography{fullBibNow}
\renewbibmacro{in:}{}
\newbibmacro{string+doi}[1]{\iffieldundef{doi}{#1}{\href{https://dx.doi.org/\thefield{doi}}{#1}}}
\DeclareFieldFormat{title}{\usebibmacro{string+doi}{\mkbibemph{#1}}}
\DeclareFieldFormat[article]{title}{\usebibmacro{string+doi}{\mkbibquote{#1}}}
\DeclareFieldFormat[incollection]{title}{\usebibmacro{string+doi}{\mkbibquote{#1}}}
\DeclareFieldFormat[inproceedings]{title}{\usebibmacro{string+doi}{\mkbibquote{#1}}}

\usepackage[pdfpagelabels, linktocpage=true]{hyperref}
\hypersetup{
colorlinks=true,
citecolor=darkblue,
linkcolor=reddish,
urlcolor=darkblue,
pdfauthor={},
pdftitle={},
pdfsubject={}
}

\long\def\ca#1\cb{} 

\newcommand{\becs}{\begin{cases}}
	\newcommand{\bem}{\begin{matrix}}

		\newcommand{\dya}[1]{|#1\rangle\langle#1|}

		\newcommand{\encs}{\end{cases}}
	\newcommand{\enm}{\end{matrix}}

\newcommand{\inp}[1]{\langle#1|#1\rangle }
\newcommand{\inpV}[2]{\langle#1,#2\rangle }
\newcommand{\inpd}[2]{\langle#1|#2\rangle }
\newcommand{\ket}[1]{|#1\rangle }

\newcommand{\ot}{\otimes }

\newcommand{\Tr}{{\rm Tr}}

\newcommand{\CC}{{\mathcal C}}
\newcommand{\DC}{{\mathcal D}}

\newcommand{\FC}{{\mathcal F}}

\newcommand{\HC}{{\mathcal H}}
\newcommand{\IC}{{\mathcal I}}
\newcommand{\JC}{{\mathcal J}}

\newcommand{\LC}{{\mathcal L}}
\newcommand{\MC}{{\mathcal M}}
\newcommand{\NC}{{\mathcal N}}
\newcommand{\OC}{{\mathcal O}}
\newcommand{\PC}{{\mathcal P}}

\newcommand{\TC}{{\mathcal T}}

\newcommand{\aB}{\textbf{a}}
\newcommand{\bB}{\textbf{b}}

\newcommand{\dB}{\textbf{d}}
\newcommand{\eB}{\textbf{e}}
\newcommand{\fB}{\textbf{f}}

\newcommand{\mB}{\textbf{m}}

\newcommand{\pB}{\textbf{p}}
\newcommand{\qB}{\textbf{q}}

\newcommand{\uB}{\textbf{u}}
\newcommand{\vB}{\textbf{v}}
\newcommand{\wB}{\textbf{w}}
\newcommand{\xB}{\textbf{x}}

\newcommand{\Rbb}{\mathbb{R}}

\newcommand{\al}{\alpha }

\newcommand{\gm}{\gamma }

\newcommand{\dl}{\delta }

\newcommand{\ep}{\epsilon}

\newcommand{\kp}{\kappa }
\newcommand{\lm}{\lambda }
\newcommand{\Lm}{\Lambda }

\newcommand{\sg}{\sigma }

\theoremstyle{remark}

\newtheorem{lemma}{Lemma}

\usepackage[linesnumbered]{algorithm2e}

\usepackage[affil-it]{authblk}

\usepackage{bbm}
\newcommand{\One}{\mathbbm{1}} 

\author[1,2]{Vikesh Siddhu}
\author[1]{John Smolin}
\title{Maximum expectation of observables with restricted purity states}
\affil[1]{IBM Quantum, IBM T.J. Watson Research Center, Yorktown Heights, NY, USA}
\affil[2]{IBM Quantum, IBM Research India}
\begin{document}

\maketitle
\begin{abstract}
    Assessment of practical quantum information processing~(QIP) remains
    partial without understanding limits imposed by noise.  Unfortunately, mere
    description of noise grows exponentially with system size, becoming
    cumbersome even for modest sized systems of imminent practical interest.
    We fulfill the need for estimates on performing noisy
    quantum state preparation, verification, and observation.  To do the
    estimation we propose fast numerical algorithms to maximize the expectation
    value of any $d$-dimensional observable over states of bounded purity.
    This bound on purity factors in noise in a measurable way.  Our fastest
    algorithm takes $O(d)$ steps if the eigendecomposition of the observable is
    known, otherwise takes $O(d^3)$ steps at worst.  The algorithms also solve
    maximum likelihood estimation for quantum state tomography with convex and
    even non-convex purity constraints.  Numerics show performance of our key
    sub-routine~(it finds in linear time a probability vector with bounded norm
    that most overlaps with a fixed vector) can be several orders of magnitude
    faster than a common state-of-the-art convex optimization solver.  Our work
    fosters a practical way forward to asses limitations on QIP imposed by
    quantum noise.  Along the way, we also give a simple but fundamental
    insight, noisy systems~(equivalently noisy Hamiltonians) always give higher
    ground-state energy than their noiseless counterparts.
\end{abstract}
\tableofcontents

\section{Introduction}

Quantum state preparation is a basic task in quantum information processing.
One may prepare a resource state to be used prior to or
for the purpose of computation~\cite{BriegelBrowneEA09}, as a vehicle for
carrying error protected quantum information~\cite{Gottesman97}, or as the end
result of a quantum computation~\cite{Shor97}.  State preparation may also be
carried out with the purpose of distributing high fidelity entanglement across
an interconnect~\cite{SiddhuSmolin23a}, a task that can enable distributed
quantum computation. Another basic task in quantum communication and
computation is the verification of a quantum state through quantum
tomography~\cite{LvovskyRaymer09, Teo15}. This method helps benchmark
components of quantum computers and verify key properties such as creation and
distribution of entanglement. Yet another task of interest is measurement of
the ground-state energy of a quantum system. The ground-state energy is not
only useful in many-body physics but also in quantum
chemistry~\cite{McArdleEndoEA20} and quantum
simulation~\cite{GeorgescuAshhabEA14}.

Quantum noise limits our ability to carry out these fundamental tasks of state
preparation, entanglement distribution, tomography, and observation of a
system's ground-state energy. Since noise is essentially inevitable, one must
understand how it limits us from performing these tasks.  Such noise can be
mathematically modelled by a quantum channel, a completely positive trace
preserving map. A lot of work has been done to study properties of quantum
channels~\cite{HolevoGiovannetti12, Wolf12}.  However, the limitations imposed
by these channels on various basic tasks of interest to quantum information
science remain to be precisely understood.  Without this understanding our
assessment of and expectation from quantum information processors remains
partial.

Knowing limitations imposed by a quantum channel on basic tasks in quantum
information is valuable, however applying these limitations in practice
requires a description of the quantum channel involved. Unfortunately, this
description grows exponentially with system size~\cite{MohseniRezakhaniEA08}
and becomes cumbersome even for modest system sizes containing tens or hundreds
of qubits. 
However, it is precisely such modest sized systems which are brewing recent
interest in quantum information processing~\cite{Preskill18} and where it is
key to asses how well these systems can perform some basic task of interest.  
One thus requires a quick, possibly rough estimate on how well a noisy system
may perform under noise.

The first problem solved in this work is to give concise bounds on the fidelity
of preparing a fixed target pure state and also the minimum energy that a
system can have, when emerging from a fixed known quantum channel. Our bound
gives a fundamental insight: noise arising from quantum channels cannot make
the ground-state energy fluctuate below its noiseless value. 

Our second set of results addresses a pertinent case where noise can't be
precisely known. In this case, we provide easily computable estimates on the
maximum fidelity of preparing a target pure state, and the minimum energy a
noisy system can have.  Our bounds take a single measurable
parameter~\cite{BarencoBerthiaumeEA96, BuhrmanCleveEA01a, WeiPritchettEA23},
the purity (more precisely its maximum value), that a noisy system can have.
We also show how to solve both convex and non-convex optimization problems of
reconstructing a quantum state in tomography while incorporating this purity
constraint.  For all these optimizations we provide two efficient
algorithms for maximizing the expectation value of a Hermitian operator over
the set of quantum states with bounded purity. If the eigendecomposition of
this, say $d \times d$, Hermitian operator is known, then our two different algorithms solve
the problem in $O(d^2)$ and $O(d)$ steps, respectively. If the decomposition
isn't known then we incur an additional~(at most $O(d^3)$) cost of obtaining
this decomposition.

The key workhorses for these optimizations are our two different algorithms that both find a
$n$-dimensional probability vector whose norm is bounded and that most overlaps with a
fixed vector. Numerics indicate our recursive algorithm and duality based
algorithm, taking $O(n^2)$ and $O(n)$ steps respectively, can be orders of
magnitude faster than CVXOPT~\cite{SraNowozinEA11, AndersenDahlEA21} for modest $n$.
We also apply our techniques to the problem of one-shot entanglement passing in
the presence of noisy state preparation. Bounds obtained in this case are shown
to be non-multiplicative.

Our work provides fundamental insights on restrictions imposed by quantum
channels on ground energies of Hamiltonians and state-of-the-art numerical
techniques for a variety of tasks in quantum information processing.
Together, the work fosters an exciting direction to incorporate the effect of
noise in one's assessment of quantum systems.

The rest of this paper is organized as follows. This next preliminaries section
begins with basic definitions and notations in Sec.~\ref{sec:qStates}, prior
technical result on one-shot entanglement sharing in Sec.~\ref{sec:priorRes},
and ends with a discussion in Sec.~\ref{sec:opt} of tools from optimization
theory required for various parts of this manuscript.
Maximum fidelity of preparing a pure state and the minimum energy a system can
take, both in the presence of a channel $\NC$, are discussed in
Sec.~\ref{sec:noisyVal}.
The next section, Sec.~\ref{sec:purLim}, formulates these same problems when
the purity of inputs is bounded from above. One-shot entanglement passing in
the purity limited setup is discussed in Sec.~\ref{sec:AdvQ}. Maximum
likelihood tomography with both fixed and upper bounded purity are formulated
in Sec.~\ref{sec:tomo}.
All these problem formulations in Secs.~\ref{sec:purLim}-\ref{sec:tomo}
are reduced from the space of matrices to the space of probability vectors;
these reductions and methods for solving them are outlined in
Sec.~\ref{sec:reformprob}.
These outlines are completely fleshed out in subsequent sections. The first
subsequent section, Sec.~\ref{sec:algSol}, contains various algebraic proofs
for various special instances of the problem at hand.
Using these special instances, a recursive solution for the general case is
given in Sec.~\ref{sec:recursiveAlgo}.
A faster duality theory based algorithm is then given in
Sec.~\ref{sec:dualAlgo}. Finally results from numerical experiments comparing
these algorithms with CVXOPT are in Sec.~\ref{sec:numerics}.
The manuscript comes to a close with a discussion and open problems in
Sec.~\ref{sec:discuss}.

\section{Preliminaries}
\label{sec:prelim}
\subsection{Quantum States and Channels}
\label{sec:qStates}

Let $\HC$ denote a $d$-dimensional Hilbert space and $\ket{\psi} \in \HC$ with
unit norm, $|\ket{\psi}|:= \inp{\psi} = 1$, be a pure state. 
Let $\LC(\HC)$ be the space of linear operators on $\HC$.  Operators $A,B \in
\LC(\HC)$ have Frobenius inner product,
\begin{equation}
    \inpV{A}{B} := \Tr (A^{\dag} B),
    \label{eq:frob}
\end{equation}
where $A^{\dag}$ is the dual~(conjugate transpose) of $A$, and $\Tr$ is trace.
A Hermitian operator $H \in \LC(\HC)$, satisfying $H = H^{\dag}$, is positive
semi-definite, $H \succeq 0$, when all its eigenvalues are non-negative.
Replacing eigenvalues of a positive semi-definite operator $H$ with
their square root results in $\sqrt{H}$.
For any $A \in \LC(H)$, the operator $1$-norm, also called the nuclear norm,
$||A||_1 := \Tr(\sqrt{AA^{\dag}})$, the operator $2$-norm, sometimes called the
Frobenius norm, $||A||_2 := \sqrt{\inpV{A}{A})}$, and the spectral norm,
sometimes called the $\infty$-norm, $||A|| := \underset{|\ket{\psi}| \leq 1}{\max}
|A \ket{\psi}|$. In general~(see Sec.1.1.3
in~\cite{Watrous18}),
\begin{equation}
    ||A||_2 \leq ||A||_1 \leq \sqrt{\text{rank}}(A)||A||_2
    \label{eq:normIneq}
\end{equation}
When $A$ is Hermitian,
\begin{equation}
    \lm_{\max}(A) = \underset{\sg \in S}{\max} \; \Tr(A\sg),
    \label{eq:infNorm}
\end{equation}
where $S:= \{\sg \in \LC(\HC) \; | \; \sg \succeq 0 \; \text{and} \; \Tr(\sg) =
1 \}$ is the set of density operators, and $\lm_{\max}(A)$ represents the
largest eigenvalue of $A$. Note $S$ is a convex set.
Using~\eqref{eq:infNorm}, the smallest eigenvalue of $A$,
\begin{equation}
    \lm_{\min}(A) = -\lm_{\max}(-A) = \underset{\sg \in S}{\min} \; \Tr(A\sg).
    \label{eq:minEig}
\end{equation}
Fidelity between a density operator $\rho$ and a pure state $\ket{\psi}$ is
\begin{equation}
    F(\rho, \psi) = \inpV{\psi}{\rho},
    \label{eq:Fidelity}
\end{equation}
where $\psi: = \dya{\psi}$.  The Fidelity between two density operators $\rho$
and $\sg$, $F(\rho,\sg) := || \sqrt{\rho} \sqrt{\sg} ||_1$.

Let $\HC_a, \HC_b$ and $\HC_r$ have equal dimension $d$ and
$\{\ket{\al_i}\}$ be an orthonormal basis~($\inpd{\al_i}{\al_j} = \dl_{ij}$) 
of a $d$-dimensional Hilbert space.
Let
\begin{equation}
    \ket{\phi}_{ra} = \frac{1}{\sqrt{d}} \sum_i \ket{\al_i}_r \ot \ket{\al_i}_a
    \label{eq:maxEnt}
\end{equation}
be a maximally entangled state on $\HC_{ra}:= \HC_r \ot \HC_a$, where $\ot$
represents tensor product, and the subscripts, $r$ and $a$, denote
the spaces, $\HC_r$ and $\HC_a$, respectively, to which these kets belong.
Let $\IC$ be the identity map taking $\LC(\HC_r)$ to itself. A linear map $\NC : \LC(\HC_a)
\mapsto \LC(\HC_b)$ is completely positive~(CP) if its Choi-Jamio\l{}kowsi
state,
\begin{equation}
    J^{\NC}_{rb} := \IC \ot \NC (\phi_{ra}),
    \label{eq:CJ}
\end{equation}
is positive semi-definite. A CP map can always be written in the Kraus form,
\begin{equation}
    \NC(A) = \sum_i K_i A K_i^{\dag}
    \label{eq:krsForm}
\end{equation}
where $K_i : \HC_a \mapsto \HC_b$ are called Kraus operators.  This
map $\NC$ is trace preserving when $\Tr_b(J^{\NC}_{rb}) = I_r/d$,
alternatively, $\sum_i K_i^{\dag} K_i = I_a$. 
A quantum channel $\NC : \LC(\HC_a) \mapsto \LC(\HC_b)$ is a completely
positive trace preserving~(CPTP) map.
Dual map $\NC^{\dag}: \LC(\HC_b) \mapsto \LC(\HC_a)$ of $\NC$ can be defined via
\begin{equation}
    \Tr\big( \NC^{\dag}(B)A \big) = \Tr \big( B\NC(A) \big).
    \label{eq:dualInp}
\end{equation}
for all $A \in \LC(\HC_a)$ and $B \in \LC(\HC_b)$. The dual map represents
evolutions of observables just as the channel represents evolution of states.
The Choi-Jamio\l{}kowsi state $J^{\NC^{\dag}}_{ra}$ on $\HC_{ra}$ can be converted to
$J^{\NC}_{rb}$ on $\HC_{rb}$ by identically mapping $\HC_a$ to $\HC_b$. This
new operator is related to the Choi-Jamio\l{}kowsi state of $\NC$ as,
\begin{equation}
    J^{\NC}_{ra} = (J^{\NC^{\dag}}_{ra})^*,
\end{equation}
where complex conjugation is done in the $\{ \ket{\al_i} \}$ basis of $\HC_r$
and $\HC_a$.

\subsection{One-shot entanglement sharing}
\label{sec:priorRes}

The aim of one-shot entanglement sharing is to prepare a state $\Lm_{ra}$ such
that sharing sub-system $a$ across $\NC$ results in a state of with maximum
fidelity to a fully entangled state $\phi_{rb}$. This maximum
fidelity~\cite{SiddhuSmolin23a, PalBandyopadhyay18, PalBandyopadhyayEA14,
VerstraeteVerschelde03},
\begin{equation}
    \OC(\NC) = \lm_{\text{max}}(J^{\NC}_{ra})
    \label{th:thOne}
\end{equation}
is achieved when input $\Lm_{ra}$ has support in the eigenspace of
$J^{\NC}_{ra}$ with largest eigenvalue. The maximum fidelity above is
multiplicative~\cite{SiddhuSmolin23a}, i.e., 
\begin{equation}
    \OC(\NC_1 \ot \NC_2) = \OC(\NC_1) \OC(\NC_2).
    \label{eq:mult}
\end{equation}
The property ensures that the maximum fidelity, $\OC(\NC)$, and
the optimum input $\Lm_{ra}$ found in the one-shot setting extend naturally
even when the channel is used multiple times.

\subsection{Optimization Theory}
\label{sec:opt}

Bold face $\xB$ denotes a vector in $n$-dimensional real space,
$\Rbb^n$ and $\xB_i$ denotes its $i^{\text{th}}$ coordinate. Rearranging these
coordinates in decreasing order results in $\xB^{\downarrow}$, a
vector satisfying $\xB^{\downarrow}_1 \geq \xB^{\downarrow}_2 \geq \dots
\geq \xB^{\downarrow}_n$. 
We use $\eB(i)$ to denote a vector whose $i^{\text{th}}$ entry is one and rest
zero, and define $\One := \sum_i \eB(i)$.
For any $\xB \in \Rbb^n$, the $l_1$-norm, $||\xB||_1 := \sum_i |\xB_i|$, and
the $l_2$-norm, $||\xB||_2 := \sqrt{\sum_i \xB_i^2}$, satisfy an inequality
of the type in eq.~\eqref{eq:normIneq},
\begin{equation}
    ||\xB||_2 \leq ||\xB||_1 \leq \sqrt{k} || \xB||_2 \leq \sqrt{n} || \xB||_2,,
    \label{eq:vecNormIneq}
\end{equation}
where $k$ is the number of non-zero entries in $\xB$.
The probability simplex in $\Rbb^n$, 
\begin{equation}
    P := \{ \pB \; | \; \pB \geq 0, \; \text{and} \; \sum_i \pB_i = 1\},
    \label{eq:prbSim}
\end{equation}
where $\pB \geq 0$ is component wise notation to specify every coordinate of
$\pB$ is non-negative. Extreme points of $P$ are $\{\eB(i)\}_{i=1}^n$.
Let 
\begin{equation}
    E_t = \{ \pB \; | \; \pB \cdot \pB \leq t \; \text{and} \; \sum_i \pB_i =
    1\},
    \label{eq:prbSet2}
\end{equation}
be a convex set of vectors with $l_2$ norm at most $\sqrt{t}$ lying in the
$\One \cdot \pB = 1$ plane.

We briefly review some basics of optimization theory~(see Ch. 4 and 5
in~\cite{BoydVandenberghe04}). Consider a convex subset of $\Rbb^n$,
\begin{equation}
    \CC = \{ \xB \; | \; h_i(\xB) \leq 0, \quad \text{and} \quad l_j(\xB) = 0\}
    \label{def:set}
\end{equation}
where $h_i: \CC \mapsto \Rbb$, $1 \leq i \leq p$, are convex functions and
$l_j: \CC \mapsto \Rbb$, $1 \leq j \leq q$, are affine. Let $f: \CC \mapsto
\Rbb$ be a concave function, define an optimization problem
\begin{equation}
    \underset{\xB \in \CC}{\text{maximize}} \; f(\xB),
    \label{def:primal}
\end{equation}
its optimum $f^{*}$ such that $f^{*} \geq f(\xB) \; \forall \xB \in \CC$, and its
optimizer(s) $\xB^* \in \CC$, such that $f^* = f(\xB^*)$. This concave
maximization problem above has an associated Lagrangian,
\begin{equation}
    \LC(\xB, \uB, \wB) := f(\xB) - \sum_i \uB_i h_i(\xB) - \sum_j \wB_j l_j(\xB),
    \label{def:Lagrange}
\end{equation}
where $\uB_i \geq 0$ and $\wB_j \in \Rbb$, and a dual objective function
\begin{equation}
    g(\uB, \wB) = \underset{\xB \in \Rbb^n}{\text{maximize}} \; \LC(\xB, \uB, \wB).
    \label{def:dualObj}
\end{equation}
This dual objective function is a convex function of $\uB, \wB$, even if the
primal objective function $f$ in~\eqref{def:primal} isn't concave~(see
Sec.~5.1.2 in~\cite{BoydVandenberghe04}) in $\xB$. Define a dual problem,
\begin{equation}
    \underset{\uB \geq 0, \wB \in \Rbb}{\text{minimize}} \; g(\uB, \wB),
    \label{def:dual}
\end{equation}
its optimum $g^* \leq g(\uB, \wB)$ for all $\uB \geq 0 \; \& \; \wB \in \Rbb$
and also define optimizers $\uB^*$ and $\wB^*$ if they satisfy $g^* = g(\uB^*, \wB^*)$.
In general, weak duality holds and thus $f^* \leq g^*$~(see Sec.~5.2.2
in~\cite{BoydVandenberghe04}).  However, if there exists $\xB$ s.t.
\begin{align}
    h_i(\xB) <0 \; \forall \; 1 \leq i \leq p \quad \text{and} \quad
    l_j(\xB) =0 \; \forall \; 1 \leq i \leq q.
    \label{eq:strictFeas}
\end{align}
then $f^* = g^*$, and one says {\em strong duality} holds~(see Sec.~5.2.3
in~\cite{BoydVandenberghe04}). 
When strong duality holds, Karush-Kunh-Tucker~(KKT) conditions,
stationarity
\begin{equation}
    \nabla_{\xB} \LC(\xB, \uB, \wB)|_{\xB, \uB, \wB} = 0,
    \label{eq:kkt1}
\end{equation}
complementary slackness,
\begin{equation}
    \uB_i h_i(\xB) = 0 \quad \forall 1 \leq i \leq p
    \label{eq:kkt2}
\end{equation}
primal feasibility,
\begin{equation}
    \xB \in \CC,
    \label{eq:kkt3}
\end{equation}
and dual feasibility,
\begin{equation}
    \uB \geq 0
    \label{eq:kkt4}
\end{equation}
are both necessary and sufficient for $\xB$ to be an optimizer
for~\eqref{def:primal} and $\uB$, $\wB$ to be optimizers
for~\eqref{def:dual}. 

In~\eqref{def:primal} if instead of maximizing $f$, one minimizes $-f$ then the
resulting problem is a convex minimization problem. This convex minimization
problem and the concave maximization problem~\eqref{def:primal} are both
commonly referred to as convex optimization problems.
In general, if $\DC$ is some convex set and $g: \DC \mapsto \Rbb$ is a convex
function we say
\begin{equation}
    \underset{A \in \DC}{\text{minimize}} \; g(A),
\end{equation}
is a convex optimization problem.

\section{Noisy Fidelity and Ground Energy}
\label{sec:noisyVal}

Consider starting with a fixed quantum state in $\HC_a$, applying some noisy
quantum operations, possibly with post-selection, with the aim to prepare a
target pure state $\psi \in \HC_b$ with highest fidelity.  If this preparation
process is fully characterized by a known quantum channel $\NC: \LC(\HC_a)
\mapsto \LC(\HC_b)$, then its highest fidelity of preparing $\psi$,
\begin{equation}
    \underset{\rho \in S_a}{\max} 
    \; F(\NC(\rho), \psi) = \lm_{\max}(\NC^{\dag}(\psi)).
\end{equation}
Proof for this statement can come from using definitions of
fidelity~\eqref{eq:Fidelity}, the dual map~\eqref{eq:dualInp}, and the
formulation of the largest eigenvalue given in~\eqref{eq:infNorm}.

Alternatively, suppose the objective is to prepare a system in a~(possibly
mixed) state with least energy with respect to some fixed Hamiltonian $H$ where
the system emerges from noise described by a quantum channel $\NC$.
This least energy,
\begin{equation}
    \underset{\rho \in S_a}{\min} \; \Tr\big( H \NC(\rho) \big) =
    \lm_{\min}(\NC^{\dag}(H)) =
    \underset{\sg \in S_b^{\NC}}{\min} \; \Tr\big( H \sg \big),
\end{equation}
where $S_a$ is the set of density operators on $\HC_a$ and $S_b^{\NC} = \{ \sg
\; | \sg \in S_b, \;  \sg  = \NC(\rho), \; \rho \in S_a\}$.
The first equality follows from the definition of the dual
channel~\eqref{eq:dualInp} and the formulation of the least eigenvalue
below~\eqref{eq:infNorm}. The second equality simply re-phrases the left side
of the first equality.
Since $S_b^{\NC}$ is a subset of $S_b$ we obtain
\begin{equation}
    \lm_{\min}(\NC^{\dag}(H)) \geq \lm_{\min}(H),
    \label{eq:minEChange}
\end{equation}
i.e., a noisy Hamiltonian, evolved under noise represented by a dual map,
always has a higher ground-state energy than the noiseless Hamiltonian.
Interestingly, the minimum energy with
respect to $H$ is obtained at some mixed state $\NC(\rho)$ at the output while
the minimum energy with respect to $\NC^{\dag}(H)$ can always be obtained at a
pure state at the channel input. These two energies are the same and thus there
is a mixed state at the channel output whose energy with respect to $H$ is the
same as the energy with respect to $\NC^{\dag}(H)$ of a pure state at the
input.

A noisy state preparation with a noiseless Hamiltonian is one way to model
noise when measuring the ground-state energy. There are other ways to model
noise and~\eqref{eq:minEChange} can constrain the results of other such
models. For instance, an alternative is to model noise in such a way that the
state preparation is noiseless but the Hamiltonian deviates from its noiseless
version $H$.  If this deviation is modelled as a perturbation, then the
perturbed Hamiltonian $H(\ep) = H + \ep V$, where $\ep$ is small and $V$ is
Hermitian.  This perturbed Hamiltonian can be used to compute an exact or
approximate new ground-state energy in the presence noise.  Our result in
eq.~\eqref{eq:minEChange} applied to this setup can constrain such a
computation to yield a ground-state energy to be strictly larger than that of
$H$.

\section{Purity Limited Fidelity and Ground Energy}
\label{sec:purLim}

An estimate of the best fidelity of preparing a target state, or the least
energy a system can have in the presence of noise, may be found using results
in Sec.~\eqref{sec:noisyVal} only when the noise channel $\NC: \LC(\HC_a)
\mapsto \LC(\HC_b)$ is known.
However, characterizing a noisy channel $\NC$ is expensive. For instance, when
$a$ and $b$ are $n$ and $m$ qubits systems, respectively, then $\NC$ has
$O(4^{n+m})$ parameters~(here and elsewhere we use the big $O$ notation). An easier
alternative is to work with a single parameter, the maximum purity $\Tr(\rho^2)
\leq t$, of the state being prepared under noise. 
If $t$ is chosen such that the inequality is satisfied for all 
$\rho \in S_b^{\NC}$, then the purity constraint becomes strictly looser than
the $\rho \in S_b^{\NC}$ constraint.
In general, this purity constraint is equivalent to 
\begin{equation}
    H_2(\rho) \geq \log \frac{1}{t},
    \label{eq:renyConst}
\end{equation}
where $H_2(\rho) := -\log \Tr(\rho^2)$ is the 2-R\'{e}nyi entropy of a density
operator. Let 
\begin{equation}
    S^t := \{\rho \; | \; \rho \succeq 0, \; \Tr(\rho) = 1 \; \text{and} \; 
    \Tr(\rho^2) \leq t \},
    \label{eq:CrSet}
\end{equation}
be a subset of $S$ containing density operators with $||\rho||_2 \leq
\sqrt{t}$.  Note $S^t$ is a convex set since $S$ is a convex set and the
Frobenius norm is a convex function of its argument.
Under the constraint $\rho \in S^t$, one formulates a convex optimization
problem to find the maximum fidelity with which one can prepare a target pure
state $\psi$,
\begin{equation}
    \underset{\rho \in S^t}{\max} \; F(\rho, \psi),
    \label{eq:maxFid}
\end{equation}
and the least energy with respect to some Hamiltonian $H$
\begin{equation}
    \underset{\rho \in S^t}{\min} \; \Tr(\rho H).
    \label{opt:maxFPhiPur}
\end{equation}
The first problem~\eqref{eq:maxFid} is a special case of the
second~\eqref{opt:maxFPhiPur}.
In principle, this second convex optimization
problem can be solved directly using well-known solvers, however the complex
semi-definite constraint on $\rho$ can make this optimization slow, specially in high
dimensions. We simplify this optimization problem in Sec.~\ref{sec:reformprob}.
Prior to discussing this simplification we address a case of special interest
where we generalize the discussion in Sec.~\ref{sec:priorRes} and discuss
quantum tomography in the section following it.

\section{Purity limited single-shot entanglement passing}
\label{sec:AdvQ}

Suppose two parties, Alice and Bob, are connected by a known channel $\NC: \LC(\HC_a)
\mapsto \LC(\HC_b)$, where $d_a = d_b = d$. In addition, Alice has access to a
reference system with Hilbert space $\HC_r$ where $d_r = d$.  Noise on
Alice's joint system $\HC_{ra}$ restricts the purity of
states $\rho_{ra}$ that Alice may prepare, $\Tr(\rho_{ra}^2) \leq t$. Given this
restriction, Alice aims to send sub-system $\HC_a$ of state $\rho_{ra}$ across
the channel $\NC$ to obtain a state $\rho_{rb}$ with highest fidelity to a
maximally entangled state.  This maximum fidelity can be found by solving the
following optimization problem:
\begin{align}
    \label{opt:maxFPhiPur-1}
    \begin{aligned}
        \text{maximize} \; & F(\rho_{rb},\phi_{rb}) \\
        \text{subject to} \; & \rho_{rb} = \IC \ot \NC (\rho_{ra}), \\
        \; & \rho_{ra} \in S^t.
    \end{aligned}
\end{align}
The optimum value of this problem is the highest fidelity,
\begin{equation}
    \OC_t(\NC):= F(\rho^*_{rb}, \phi_{rb}),
    \label{def:maxFIDR}
\end{equation}
where $\rho_{ra}^*$ is the optimizer in~\eqref{opt:maxFPhiPur-1}.

Optimizer $\Lm_{ra}$ to~\eqref{opt:maxFPhiPur-1} without the purity constraint
$\Tr (\rho_{ra}^2) \leq t$ is given in~\eqref{th:thOne}.  Using $\Lm_{ra}$
one may guess a solution for~\eqref{opt:maxFPhiPur-1}. One guess is a weighted
sum of $\Lm_{ra}$ with a maximally mixed state.  
While this guess is correct in special cases~(see Sec.~\ref{sec:simplEx}), the
general solution to~\eqref{opt:maxFPhiPur-1} can be different. To find this
solution, one reformulates the problem.
Using the definitions of fidelity~\eqref{eq:Fidelity}, dual
channel~\eqref{eq:dualInp}, and Choi-Jamio\l{}kowsi operator~\eqref{eq:CJ}, the
optimization problem~\eqref{opt:maxFPhiPur-1} can be re-written as
a special case of~\eqref{opt:maxFPhiPur},
\begin{equation}
    \underset{\rho \in S^t}{\text{maximize}} \; 
    \frac{1}{d} \Tr(\rho \JC^{\NC^{\dag}}_{ra}).
    \label{opt:maxFPhiPur-2}
\end{equation}

\subsection{Example Channels}
\label{sec:simplEx}

In certain simple cases, the optimization problem~\eqref{opt:maxFPhiPur-2} can
be solved algebraically~(see Sec.~\ref{sec:basisVec} for details). The first
case is when $\NC = \IC_d$, the identity channel on a $d$-dimensional system.
In this case $\JC_{ra}^{\NC^{\dag}} = \phi_{ra}$,
\begin{equation}
    O_t(\IC_d) = \frac{1}{d^2}(1 + \sqrt{(td^2 - 1)(d^2 - 1)})
\end{equation}
the optimizer,
\begin{equation}
    \rho_{ra}^* = (1 - \lm) \frac{I}{d^2} + \lm \phi^{+}_{ra}
\end{equation}
where $\lm = \sqrt{\frac{td^2 - 1}{d^2 - 1}}$.
A second case is $\NC(\rho) = \TC_d(\rho) := \Tr(\rho) \dya{0}$, sometimes called
the trace channel. In this case,
\begin{equation}
    \JC_{ra}^{\NC^{\dag}} = \dya{0}_{r} \ot \frac{I_a}{d},
\end{equation}
\begin{equation}
    O_t(\TC_d) = 
    \begin{cases}
       \frac{1}{d} & \text{if} \quad \frac{1}{d} \leq t \leq 1 \\
        \frac{1}{d^2}(1 + \sqrt{(td^2-1)(d-1)}) & \text{if} \quad \frac{1}{d^2}
        \leq t \leq \frac{1}{d}
    \end{cases},
        \label{eq:TrSol}
\end{equation}
and the optimizer
\begin{equation}
    \rho_{ra}^* = 
    \begin{cases}
        \dya{0}_r \ot \rho_{a} & \text{if} \quad \frac{1}{d} \leq t \leq 1 \\
        (1-k) \frac{I_{ra}}{d^2} + k \dya{0}_r \ot \frac{I}{d} & \text{if} \quad \frac{1}{d^2}
        \leq t \leq \frac{1}{d}
    \end{cases},
\end{equation}
where $\rho_a$ is any density operator with $\Tr(\rho_a^2) \leq t$ and $k = \lm
\sqrt{d+1} =  \sqrt{\frac{td^2 - 1}{d - 1}}$.
A third case is $\NC = \IC_d \ot \TC_d$, the tensor product of an identity
channel and trace channel, each acting on a $d$-dimensional input. Here 
$\IC: \LC(\HC_{a1}) \mapsto \LC(\HC_{b1})$,
$\TC: \LC(\HC_{a2}) \mapsto \LC(\HC_{b2})$,
$\NC: \LC(\HC_{a}) \mapsto \LC(\HC_{b})$, 
$\HC_a := \HC_{a1} \ot \HC_{a2}$,
$\HC_b := \HC_{b1} \ot \HC_{b2}$,
and
$\HC_r := \HC_{r1} \ot \HC_{r2}$.
In this case
\begin{equation}
    \JC^{\NC^{\dag}}_{ra} = \phi^+_{r1a1} \ot \dya{0}_{r2} \ot \frac{I_{a2}}{d},
\end{equation}
\begin{equation}
    \OC_{r}(\IC_d \ot \NC_d) = 
    \begin{cases}
        1/d & \text{if} \quad \frac{1}{d} \leq t \leq 1\\
        \frac{1}{d^4}\big(1 + \sqrt{(td^4-1)(d^2-1)} \big) & \text{if} \quad
        \frac{1}{d^4} \leq t \leq \frac{1}{d} 
    \end{cases}
\end{equation}
and the optimizer
\begin{equation}
    \rho_{ra} = 
    \begin{cases}
        \phi^+_{r1a1} \ot \dya{0}_{r2} \ot \rho_{a2} & 
        \text{if} \quad \frac{1}{d} \leq t \leq 1 \\
        (1-\gm) \frac{I}{d^4} + \gm \phi^+_{r1a1} \ot \dya{0}_{r2} \ot
        \frac{I_{a2}}{d} & \text{if} \quad \frac{1}{d^4} \leq t \leq
        \frac{1}{d}
    \end{cases}
\end{equation}
where $\gm = \sqrt{\frac{td^4 - 1}{d^3 - 1}}$, and $\rho_{a2}$ is any density
operator with $\Tr(\rho_{a2}^2) \leq t$ when $1/d \leq t \leq 1$. 

\subsection{Multiplicativity}
\label{sec:twoCCase2}

In various practical scenarios, one not only has access to a single use of a
channel but has access to multiple uses of possibly different channels.  The
difference can arise even when one is using a fixed physical channel, but its
characteristics change over time.  It is natural to consider joint uses of these
channels to obtain high fidelity fully entangled states. In
the simplest scenario channels are used multiple times, but no additional
post-processing is done and one assumes noiseless state preparation.
As mentioned in the discussion containing eq.~\eqref{eq:mult}, the maximum
fidelity is multiplicative in this case.

A next simplest case is a limitation on the purity of the state being prepared
prior to transmission. In this case, one is interested in comparing
the maximum fidelity $\OC_t$ for two different channels used
together.
To do this comparison, one is interested in inequalities that relate $\OC_t$
computed across multiple channels with similar computations done for individual
channels.  The simplest inequality of this type, involving two channels $\NC_1
: \LC(\HC_{a1}) \mapsto \LC(\HC_{b1})$ and $\NC_2 : \LC(\HC_{a2}) \mapsto
\LC(\HC_{b2})$ is
\begin{equation}
    \OC_t(\NC_1 \ot \NC_2) \geq \OC_t(\NC_1) \OC_t(\NC_2).
    \label{eq:N1N2}
\end{equation}
It is straightforward to see that the inequality~\eqref{eq:N1N2} holds. Let
\begin{equation}
    \OC_t(\NC_1) = F(\rho_{r1b1}^*, \phi_{r1b1}) \quad \text{and} \quad
    \OC_t(\NC_2) = F(\rho_{r2b2}^*, \phi_{r2b2})
\end{equation}
where $\rho_{r1b1}^*$ and $\rho_{r2b2}^*$ are both valid density operators with
$\Tr\big( (\rho^*_{r1a1})^2 \big) \leq t$ and $\Tr\big( (\rho^*_{r2a2})^2 \big)
\leq t$. Notice $\rho_{ra}:= \rho_{r1a1}^* \ot \rho_{r2a2}^*$ is a valid
density operator with $\Tr(\rho_{ra}^2) \leq t$ and
\begin{equation}
    \OC_t(\NC_1 \ot \NC_2) \geq F(\rho_{rb}, \phi^+_{rb}) =  \OC_t(\NC_1) \OC_t(\NC_2)
    \label{eq:ineqt}
\end{equation}
where the final equality makes use of the fact that $\phi^+_{ab} =
\phi^+_{a1b1} \ot \phi^+_{a2b2}$. 
We find that the inequality~\eqref{eq:ineqt} is strict for the simplest of
channels. For instance, let $\NC_1 = \NC_2 = \IC_d$, an identity channel acting
on a $d$-dimensional Hilbert space. Then, using the results in
Sec.~\eqref{sec:simplEx} one finds that
\begin{equation}
    \OC_t(\IC_d \ot \IC_d) = \frac{1}{d^4}(1 + \sqrt{(td^4 - 1)(d^4 - 1)})
    \quad \text{and} \quad
    \OC_t(\IC_d)^2 = \frac{1}{d^4}\big(1 + \sqrt{(td^4-1)(d^4-1)}\big)^2
\end{equation}
such that inequality
\begin{equation}
    \OC_t(\IC_d \ot \IC_d) \geq \OC_t(\IC_d) \OC_t(\IC_d)
\end{equation}
at $d=2$ is strict for all $1/d^2 \leq t < 1$ and becomes an equality at
$t=1$~(see Fig.~\ref{fig:Id}).

\begin{figure}[ht]
    \centering
    \begin{subfigure}[b]{0.4\textwidth}
         \centering
         \includegraphics[width=0.8\textwidth]{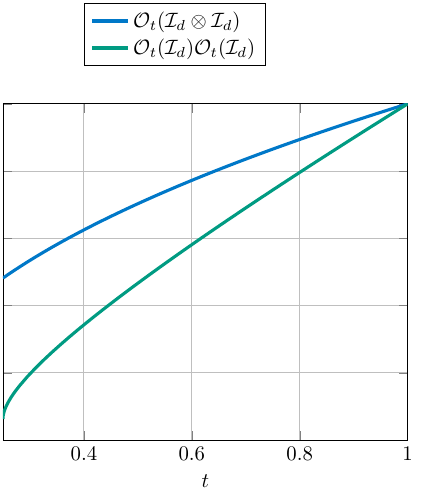}
         \caption{Plot of $\OC_t$ for two joint channel uses and two separate
         uses of the same identity channel.}
     \end{subfigure}
     \hspace{1cm}
    \begin{subfigure}[b]{0.4\textwidth}
         \centering
         \includegraphics[width=0.8\textwidth]{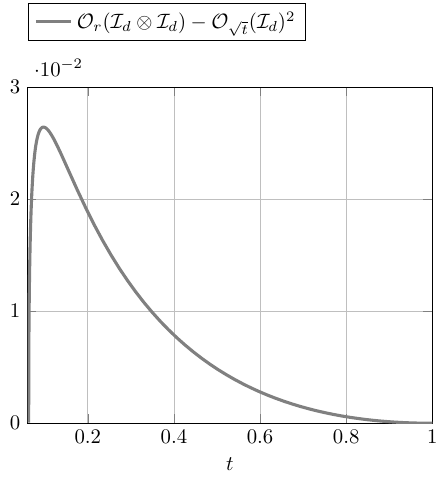}
         \caption{Plot comparing $\OC_t$ over two channel uses with
         $\OC_{\sqrt{t}}$ over a single use of each channel.}
     \end{subfigure}
    \caption{Non-multiplicative behavior with the identity channel $\IC_d$ with
    $d=2$.}
    \label{fig:Id}
\end{figure}

Similarly, if $\NC_1 = \NC_2 = \Tr$, the trace channel. Then using results
from Sec.~\ref{sec:simplEx} we find that the inequality,
\begin{equation}
    \OC_t(\TC_d \ot \TC_d) \geq \OC_t(\TC_d) \OC_t(\TC_d),
\end{equation}
at $d=2$ is strict when $1/d^2 \leq t < 1/d$ and becomes an equality at $1/d
\leq t \leq 1$~(see Fig.~\ref{fig:Tr}). Finally, we consider the inequality
\begin{equation}
    \OC_t(\IC_d \ot \TC_d) \geq \OC_t(\IC_d) \OC_t(\TC_d).
\end{equation}
At $d=2$ this inequality is strict for all $1/d^2 \leq t < 1$~(see
Fig.~\ref{fig:IdTr}). 

\begin{figure}[ht]
    \centering
    \begin{subfigure}[b]{0.4\textwidth}
         \centering
         \includegraphics[width=0.8\textwidth]{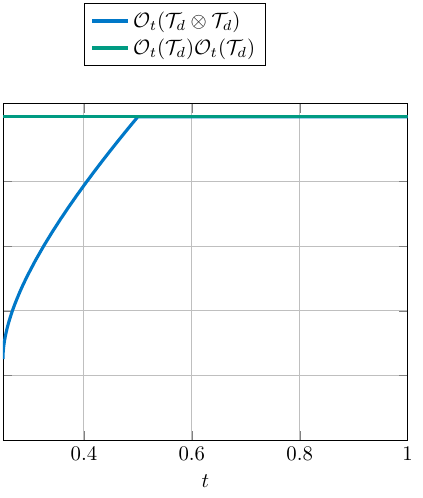}
         \caption{Plot of $\OC_t$ for two joint channel uses and two separate uses
         of the same trace channel.}
     \end{subfigure}
     \hspace{1cm}
    \begin{subfigure}[b]{0.4\textwidth}
         \centering
         \includegraphics[width=0.8\textwidth]{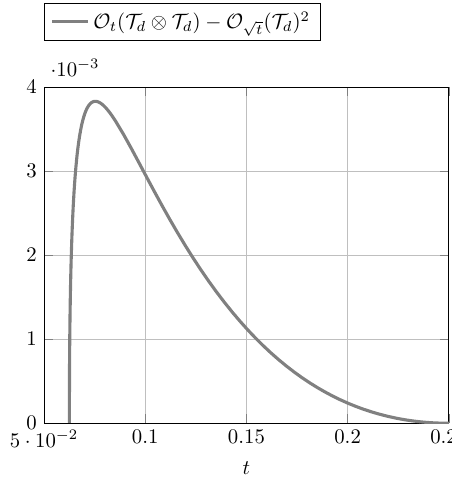}
         \caption{Plot comparing $\OC_t$ over two channel uses with
         $\OC_{\sqrt{t}}$ over a single use of each channel.}
     \end{subfigure}
    \caption{Non-multiplicative behavior with the trace channel $\Tr_d$ with $d=2$.}
    \label{fig:Tr}
\end{figure}

A second inequality of the type in~\eqref{eq:N1N2} is
\begin{equation}
    \OC_t(\NC_1 \ot \NC_2) \geq \OC_{\sqrt{t}}(\NC_1) \OC_{\sqrt{t}}(\NC_2),
    \label{eq:N1N2Rtr}
\end{equation}
which follows from the observation that if optimizers $\rho_{r1a1}^* \in
S^{\sqrt{t}}$ and $\rho_{r2a2}^* \in S^{\sqrt{t}}$ then their product
$\rho_{r1a1}^* \ot \rho_{r2a2}^* \in S^t$.

The right side of the above equation represents the maximum fidelity with a
fully entangled state at the output of two channels when using a product
channel input, such that each input has purity at most $\sqrt{t}$.  The left
side represents this same maximum fidelity, however the joint channel input can
now have correlations but the joint input itself has purity at most $t$.
It is worth asking if the left side can be larger than the right side, i.e.,
does the overlap with a fully entangled state increase when uncorrelated
channel inputs, each with purity at most $\sqrt{t}$, are replaced with
correlated inputs with purity at most $t$?

We find the inequality~\eqref{eq:N1N2Rtr} can be strict. If $\NC_1 = \NC_2 =
\IC_d$, the identity channel, then $\NC_1 \ot \NC_2 =
\IC_{d^2}$, and the inequality
\begin{equation}
    \OC_r(\IC_d \ot \IC_d) \geq \OC_{\sqrt{r}}(\IC_d) \OC_{\sqrt{r}}(\IC_d).
\end{equation}
is strict at $d=2$ for all $\frac{1}{d^4}<r<1$~(see Fig.~\ref{fig:Id}).
When $\NC_1 = \NC_2 = \TC_d$ then the inequality
\begin{equation}
    \OC_r(\TC_d \ot \TC_d) \geq \OC_{\sqrt{r}}(\TC_d) \OC_{\sqrt{r}}(\TC_d)
\end{equation}
at $d=2$ remains strict for $1/d^4 < t < 1/d^2$ and becomes an equality for 
$1/d^2 \leq t \leq 1$ and at $t = 1/d^4$~(see Fig.~\ref{fig:Tr}).
Lastly, we consider the inequality,
\begin{equation}
    \OC_t(\IC_d \ot \TC_d) \geq \OC_{\sqrt{t}}(\IC_d) \OC_{\sqrt{t}}(\TC_d)
\end{equation}
which is strict at $d=2$ for all $1/d^4 < t < 1$~(see Fig.~\ref{fig:IdTr}).

\begin{figure}[ht]
    \centering
    \begin{subfigure}[b]{0.4\textwidth}
         \centering
         \includegraphics[width=0.8\textwidth]{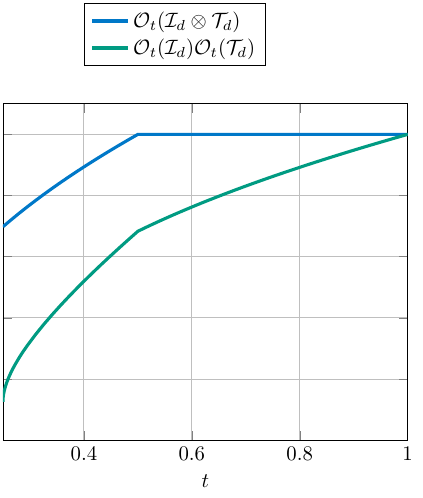}
         \caption{Plot of $\OC_t$ for joint and separate
         uses of the identity and trace channels.}
     \end{subfigure}
     \hspace{1cm}
    \begin{subfigure}[b]{0.4\textwidth}
         \centering
         \includegraphics[width=0.8\textwidth]{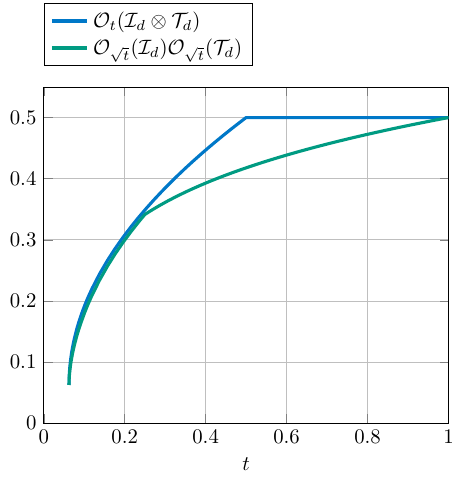}
         \caption{Plot comparing $\OC_t$ for joint use of identity and trace channel
         with $\OC_{\sqrt{t}}$ for separate use of each channel.}
     \end{subfigure}
    \caption{Non-multiplicative behavior of using Trace and identity channel
        with $d=2$}
    \label{fig:IdTr}
\end{figure}

Another application of the limited purity set up is discussed next.

\section{Purity Limited Tomography}
\label{sec:tomo}

Consider the case in quantum state tomography where identically prepared copies
of an unknown state $\rho \in \LC(\HC_a)$ are measured to construct a
(classical)~description of $\rho$.  Typically, preparation of $\rho$ and its
measurement are both noisy, in addition the number of measurements is finite.
The noise and finite nature of measurement data makes constructing a classical
description of $\rho$ non-trivial.

Mathematically, measurements can be described using a positive operator-valued
measure~(POVM) $\{\Lm_i\}$, a collection of positive semi-definite operators
summing to the identity.
Let $\{\Lm_i\}$ form a basis of $\LC(\HC_a)$. In a given of tomography on a
state $\rho$ using POVM $\{Lm_i\}$, let $m_i$ denote the number of times
the outcome $i$ corresponding to POVM element $\Lm_i$ is observed, and $m =
\sum_i m_i$ be the total number of measurements. Choose $\hat p_i = m_i/m$ to
be an estimate of the probability of observing outcome $i$. Solving a system of
linear equations,
\begin{equation}
    \Tr(\rho \Lm_i) = \hat p_i,
    \label{eq:linConst}
\end{equation}
results in the well-known linear inversion estimate of $\rho$,
\begin{equation}
    H = \sum_i \hat p_i \Lm_i^D
    \label{eq:linInv}
\end{equation}
where $\{\Lm_i^D\}$ are the dual basis corresponding to $\{\Lm_i\}$ satisfying
\begin{equation}
    \inpV{\Lm_i^D}{\Lm_j} := \dl_{ij}.
\end{equation}
A straightforward generalization of the above procedure can be obtained when
POVM $\{\Lm_i\}$ are replaced with a Hermitian basis of operators.
The linear inversion estimate in~\eqref{eq:linInv} has unit trace but it need not
be positive semi-definite.
However, an estimate that is both unit trace and positive semi-definite can be
obtained from the linear inversion estimate~\eqref{eq:linInv}. This maximum
likelihood estimate~(MLE) of $\rho$~\cite{SmolinGambettaEA12a} is the optimum
to the convex optimization problem
\begin{equation}
    \underset{\rho \in S}{\text{minimize}} \; || \rho - H||_2^2.
    \label{eq:MLE}
\end{equation}

One may measure the purity of the state $\rho$~(for instance by doing a SWAP
test on two identical copies of a system~\cite{BarencoBerthiaumeEA96,
BuhrmanCleveEA01a} or a randomized benchmarking
experiment~\cite{WeiPritchettEA23}), however the MLE, as stated
in~\eqref{eq:MLE}, can't naturally take this purity into account. For instance,
appending a purity constraint to~\eqref{eq:linConst} results in a non-linear
equation in $\rho$. We propose modifications that take the purity into account.

Suppose the measured purity is $t$ then a MLE estimate which takes this purity
into account can be constructed in two different way. The first method is to
carry out the optimization in~\eqref{eq:MLE} with the additional constraint
$\Tr(\rho^2) = t$, i.e.,
\begin{equation}
    \underset{\rho \in S, \Tr(\rho^2) = t}{\text{minimize}} \; || \rho - H||_2
    \label{eq:purityConst1}
\end{equation}
The optimization above is non-convex, even though the objective function
$||\rho-H||_2$ is convex in $\rho$, the set of states $\rho \in S$ with
$\Tr(\rho^2) = t$ is not convex. However, we show in Sec.~\ref{sec:reformprob}
this non-convex problem can be solved optimally and efficiently using our
algorithm.
The optimum found gives an estimate $\hat \rho_1$ whose purity is exactly $t$.

It is also possible to obtain an estimate with purity at most $t$.  One can
carry out a second procedure by optimizing the objective function
in~\eqref{eq:MLE} with the additional constraint $\Tr(\rho^2) \leq t$, i.e.,
\begin{equation}
    \underset{\rho \in S^t}{\text{minimize}} \; || \rho - H||_2
    \label{eq:purityConst2}
\end{equation}
The convex optimization problem above can solved using a duality based
algorithm discussed in Sec.~\ref{sec:reformprob}.

\section{Optimizing over probability vectors}
\label{sec:reformprob}

\subsection{Maximizing expectation value}
We simplify the optimization problem~\eqref{opt:maxFPhiPur} where $H$ is
a $d \times d$ Hermitian matrix.
Consider the singular value decomposition of $H$ in~\eqref{opt:maxFPhiPur},
\begin{equation}
    H = U D U^{\dag},
    \label{eq:choiEig}
\end{equation}
where $U$ is unitary and $D$ is diagonal; the $i$-th diagonal entry of $D$ is
an eigenvalue of $H$ with corresponding eigenvector forming the $i$-th column
of $U$.
Obtaining the eigendecomposition takes at most $O(d^3)$ steps or fewer using
more advances algorithms~\cite{PanChen99, DemmelDumitriuEA07}.
Let 
\begin{equation}
    \sg := U^{\dag} \rho U,
    \label{eq:newVar}
\end{equation}
it is easy to check that $\sg \in S^t$ and $\rho \in S^t$ imply
one another. As a result, re-write~\eqref{opt:maxFPhiPur} as
\begin{equation}
    \underset{\sg \in S^t}{\text{maximize}} \; \Tr(\sg D),
    \label{opt:maxFPhiPur-3}
\end{equation}
In this re-writing, restricting $\sg$ to be diagonal~($[\sg]_{ij} = 0$ if $i
\neq j$) does not change the optimum value.
To prove this claim, note adding this restriction to the maximization
problem~\eqref{opt:maxFPhiPur-3} results in an optimization problem with a
possibly lower optimum value.  However, any optimizer $\sg^*$
to~\eqref{opt:maxFPhiPur-3} can be turned into a diagonal operator,
\begin{equation}
    \sg' = \MC(\sg^*) = \sum_j \dya{j} \sg^* \dya{j}
    \label{eq:mChan}
\end{equation}
such that $\Tr(\sg^*D) = \Tr(\sg'D)$ and $\sg' \in S^t$. This last relation,
$\sg^* \in S^t \implies \sg' \in S^t$ can be shown as follows.  Notice, $\sg'$
is obtained by passing $\sg^*$ through a measurement channel $\MC$ with Kraus
operators $\{ \dya{i} \}$.  Thus $\sg^*$ is a density operator implies $\sg'$
is also a density operator. This channel $\MC$ satisfies a pinching
inequality~(see pg.~97 in~\cite{Bhatia97}) and thus,
\begin{equation}
    ||\sg'||_2 \leq ||\sg^*||_2 .
    \label{eq:pinch1}
\end{equation}
The above inequality along with the fact that $\sg^* \in S^t \implies
||\sg^*||_2 \leq \sqrt{t}$ completes the proof.

A diagonal operator $\sg \in S^t$ can be represented using a real vector
$\pB$ whose $i^{\text{th}}$ entry is simply $[\sg]_{ii}$. This
$d$-dimensional vector must have non-negative entries, $\pB \geq 0$, and its
inner product with itself must be less than $t$, i.e, $\pB \cdot \pB \leq t$.
The diagonal operator $D$ in~\eqref{eq:choiEig} can be written as $D =
\text{diag} (\qB)$, where $\qB$ has length $d$ and the notation $\text{diag}(\aB)$
represents a diagonal matrix whose $(i,j)$ entry is $\dl_{ij}\aB_i$.
Using these vectors $\pB$ and $\qB$, the optimization in~\eqref{opt:maxFPhiPur-3}
becomes equivalent to
\begin{equation}
    \underset{\pB \in P(t)}{\text{maximize}} \; \pB \cdot \qB
    \label{opt:maxFPhiPur-4}
\end{equation}
where
\begin{equation}
    P(t) := \{ \pB \; | \; \pB \geq 0, \; \sum_i \pB_i = 1,\; \text{and} \;
    \pB.\pB \leq t\}.
    \label{eq:prbSet}
\end{equation}
Using arguments similar to those below~\eqref{eq:CrSet}, one can show the
optimization problem in~\eqref{opt:maxFPhiPur-4} to be convex.
Once the optimum value of~\eqref{opt:maxFPhiPur-4} is known, the optimum value
of~\eqref{eq:choiEig}, which equals that of~\eqref{opt:maxFPhiPur-4}, is also
known with no additional effort. If required, the optimizer
of~\eqref{eq:choiEig} can be constructed from that of~\eqref{opt:maxFPhiPur-4}
with at most $O(d^3)$ additional steps.
If $\pB^*$ is an optimizer of~\eqref{opt:maxFPhiPur-4} then $\sg^* =
\text{diag}(\pB^*)$ is an optimizer of~\eqref{opt:maxFPhiPur-3} and $\rho^* =
U^{\dag} \sg^* U$ is an optimizer of~\eqref{opt:maxFPhiPur}. This reconstruction
of $\rho^*$ from $\pB^*$ can always be done using ordinary matrix multiplication
in $O(d^3)$ operations, or fewer operations using more efficient methods.

The optimization problem in~\eqref{opt:maxFPhiPur-4} can be solved numerically
using a direct recursive algorithm outlined as Alg.~\eqref{alg:recursiveShort}.
This recursive algorithm gives the optimum and optimizer
of~\eqref{opt:maxFPhiPur-4} in $O(d^2)$ steps~(see comments
below~\eqref{eq:prbSet} for obtaining optimum and optimizer
of~\eqref{eq:choiEig}). A proof for its correctness and the algorithms analysis
are in Sec.~\ref{sec:recursiveAlgo}. 

\RestyleAlgo{boxruled}
 \begin{algorithm}[H]
 \DontPrintSemicolon
 \SetAlgoLined
\BlankLine
     $\bar{\pB} :=$ $\pB$ pushed towards boundary of $E_t$~\eqref{eq:prbSet2}
\BlankLine
     \uIf{$t \geq 1$ or $t \leq 1/(d-1)$ or $\bar{\pB} \geq 0 $ or $\qB = \One/n$}{
     \Return{Appropriate algebraic solution} }
\BlankLine
 \Else{
     $\bar{\qB}$ := $\qB$ pushed to boundary $P'$ of $P$\;
     \Return{\eqref{opt:maxFPhiPur-4} solved on $P'$ with $\qB \rightarrow \bar{\qB}$}
 }
     \caption{Outline of recursive algorithm~\eqref{alg:recursive}}
     \label{alg:recursiveShort}
\end{algorithm}

The problem in~\eqref{opt:maxFPhiPur-4} can also be solved via a convex
one-parameter minimization outlined in Alg.~\eqref{alg:dualOutline}. The full
algorithm, its proof for correctness using Lagrange duality, and
its analysis are in Sec.~\ref{sec:dualAlgo}.  This algorithm gives the optimum
and optimizer of~\eqref{opt:maxFPhiPur-4} in $O(d)$ steps. We also show the
optimizer given by this algorithm can always be chosen such that its purity is
exactly $t$~(see Lemma~\ref{lem:purityEq}). As a result, the optimizer found
this way is also an optimizer for the non-convex problem
\begin{equation}
    \underset{\pB \in P, \; \pB \cdot \pB = t}{\text{maximize}} \; \pB \cdot \qB.
    \label{opt:maxFPhiPur-6}
\end{equation}

\RestyleAlgo{boxruled}
 \begin{algorithm}[H]
 \DontPrintSemicolon
 \SetAlgoLined
 \uIf{$t \geq 1$ or $t \leq 1/d$}{
     \Return{Algebraic Solution} }
\BlankLine
 \Else{Minimize a one-parameter function \;
     Construct optimizer to~\eqref{opt:maxFPhiPur-4}\;
     \Return{Optimum, Optimizer}}
     \caption{Outline of dual algorithm}
     \label{alg:dualOutline}
\end{algorithm}

\subsection{Minimizing distance}
\label{sec:minD}

Optimization problems~\eqref{eq:purityConst1} and \eqref{eq:purityConst2} can
be simplified. In the first problem~\eqref{eq:purityConst1}, replace the
objective function with its square,
\begin{equation}
    ||\rho - H||_2^2 = \Tr(\rho^2) + \Tr(H^2) - 2 \Tr(\rho H),
    \label{eq:objFun}
\end{equation}
use the constraint $\Tr(\rho^2) = t$, to arrive at a new problem,
\begin{equation}
    \underset{\rho \in S, \Tr(\rho^2) = t}{\text{maximize}} \; \Tr(\rho H),
\end{equation}
with the same optimizer as~\eqref{eq:purityConst1}.
The optimum of the above problem is bounded from above by that
of~\eqref{opt:maxFPhiPur}.  This upper bound can be found using
Algorithm~\ref{alg:dualOutline} outlined previously.  The algorithm can return
an optimizer whose purity is exactly $t$~(see discussion above
eq.~\eqref{opt:maxFPhiPur-6}). Consequently the upper bound found this way is
tight and solves~\eqref{eq:purityConst1}. The optimization performed this way
has identical performance to that of Algorithm~\ref{alg:dualOutline} being used
to solve~\eqref{opt:maxFPhiPur}.

Optimizer of the second problem~\eqref{eq:purityConst2} can be obtained by
minimizing the function to the right of~\eqref{eq:objFun}.
Minimization of this function can be re-written using~\eqref{eq:choiEig}
and~\eqref{eq:newVar} as
\begin{equation}
    \underset{\sg \in S^t}{\text{minimize}} \; || \sg - D||^2_2.
    \label{eq:minTwo}
\end{equation}
Using the definition of $\MC$ in~\eqref{eq:mChan} and the pinching
inequality~(see pg.~97 in~\cite{Bhatia97}) we find
\begin{equation}
    || \MC(\sg - D) ||_2 \leq || \sg - D||_2
\end{equation}
and
\begin{equation}
    \MC(\sg - D) = \sg'- D 
\end{equation}
where $\sg'$ is diagonal in the standard basis. Using arguments similar to
those below~\eqref{eq:pinch1} we find $\sg' \in S^t$. Let $\sg' = \text{diag} (\pB)$,
$D = \text{diag} (\qB)$ then it follows that the minimization in~\eqref{eq:minTwo}
can be written as
\begin{equation}
    \underset{\pB \in P(t)}{\text{minimize}} \; || \pB - \qB||^2_2.
    \label{eq:minThree}
\end{equation}
The optimizer $\pB^*$ of the above optimization can be turned into an optimizer
$\rho^* = U^{\dag} \text{diag}(\pB^*) U$ of~\eqref{eq:purityConst2}.
The above minimization can be re-written as
\begin{equation}
    \underset{k \leq t}{\text{minimize}} \; 
    \{ \underset{\pB \in P, \pB \cdot \pB = k}{\text{minimize}} \; 
    || \pB - \qB||^2_2 \}.
    \label{eq:minFour}
\end{equation}
The inner minimization is equivalent to maximization~\eqref{opt:maxFPhiPur-6}.
That maximization can be carried out using Algorithm~\eqref{alg:dualOutline}.
The outer minimization in~\eqref{eq:minFour} is of a one-parameter function
over the interval $1/d \leq k \leq t$. This one parameter function over a
bounded interval can be minimized with relative ease. Numerics indicate this
function to be convex in its argument $k$. Such convexity can make the
minimization even easier using algorithms such as golden section
search~(see~\cite{PressTeukolskyEA92}).  
This algorithm shrinks the search space by a constant fraction in each iteration.
Such shrinking exponentially narrows the search space with each iteration. 

\section{Algebraic Solutions}
\label{sec:algSol}

We discuss a variety of algebraic solutions to the optimization
problem~\eqref{opt:maxFPhiPur-4}. The discussion uses the notation $\pB \in \Rbb^n$
instead of $\Rbb^d$ as in~\eqref{opt:maxFPhiPur-4}. The new notation helps
keep this section self-contained. We discuss properties of various
sets related to $P(t)$, defined in eq.~\eqref{eq:prbSet},
and a simplification of the objective function in~\eqref{opt:maxFPhiPur-4}.
In particular,
The optimization problem~\eqref{opt:maxFPhiPur-4} can be solved by focussing
on $\qB$ with $\One \cdot \qB = 1$. Let $s:= \sum_i \qB_i$, and define,
\begin{equation}
    \qB' := \qB + \One \frac{1-s}{n},
    \label{eq:qbPrime}
\end{equation}
a vector whose entries sum to one. For any $\pB \in P(t)$,
$\pB \cdot \qB' = \pB \cdot \qB + (1-s)/n$ and thus
\begin{equation}
    \underset{\pB \in P(t)}{\text{argmax}} \; \pB \cdot \qB' = 
    \underset{\pB \in P(t)}{\text{argmax}} \; \pB \cdot \qB.
    \label{opt:qbPrime}
\end{equation}
Thus, we restrict ourself to $\qB$ satisfying $\One \cdot \qB = 1$.

\subsection{Special subsets of $\Rbb^n$}

We are interested in three subsets of $\Rbb^n$, the probability simplex
$P$ in~\eqref{eq:prbSim}, the set $E_t$ in~\eqref{eq:prbSet2} and their
intersection, $P(t)$ in~\eqref{eq:prbSet}. Various properties of these sets
are useful in finding algebraic solutions to~\eqref{opt:maxFPhiPur-4}.
The Lemmas below capture properties of interest.

\begin{lemma}
    When $t \geq 1$, the set $P(t)$ is the same as $P$.
    \label{eq:tGeq1Lemm}
\end{lemma}
\begin{proof}
    Follows from the definitions of $P$, $P(t)$, and
    inequality~\eqref{eq:vecNormIneq}.  Namely, $\pB \in P$ has $||\pB||_2 \leq
    ||\pB||_1 = 1$ and thus $||\pB||_2 \leq \sqrt{t}$ is not necessary when
    $t\geq 1$ and $P(t)$ is simply $P$.
\end{proof}

\begin{lemma}
    When $t<1/n$, the set $E_t$ and $P(t)$ are both empty.
    \label{lem:EmptyFt}
\end{lemma}
\begin{proof}
    Any $\pB \in E_t$ satisfies $ \One \cdot \pB = 1 \leq ||\pB||_1$, using
    this equality together with~\eqref{eq:vecNormIneq} we get 
    \begin{equation}
        \frac{1}{\sqrt{n}} \leq || \pB ||_2,
        \label{eq:NormBound}
    \end{equation}
    for all $\pB \in E_t$. However any $\pB \in E_t$ has $||\pB||_2 \leq
    \sqrt{t} < 1/\sqrt{n}$~(see statement of this Lemma). Thus $E_t$ is empty
    and the intersection $P(t)$ of $P$ and $E_t$ is also empty.
\end{proof}

\begin{lemma}
    At $t=1/n$ the sets $E_t = P(t) = \{\One/n\}$.
    \label{lem:SingleEt}
\end{lemma}
\begin{proof}
    Any $\pB \in E_t$ satisfies $\pB \cdot \pB \leq t = 1/n$. This inequality
    together with~\eqref{eq:NormBound} gives $|| \pB ||_2 = 1/\sqrt{n}$. 
    In the Cauchy-Schwarz inequality,
    \begin{equation}
        |\uB \cdot \vB| \leq ||\uB||_2 ||\vB||_2,
        \label{eq:CS}
    \end{equation}
    where equality holds if and only if $\uB = \lm \vB$ for some $\lm \in
    \Rbb$, let $\uB = \pB \in E_t$ and $\vB = \One/n$. Then equality holds and
    $\pB = \lm \One/n$ for $\lm \in \Rbb$, however $\pB \in E_t$ implies $\One
    \cdot \pB = 1$ and thus $\lm = 1$. Since $E_t = \{\One/n\}$ contains a
    single element which is a probability vector, intersection of $E_t$ with
    the probability simplex, $P(t)$, is just $\{\One/n\}$
\end{proof}

\begin{lemma}
    The set $P(t)$ equals $E_t$ when $1/n \leq t \leq 1/(n-1)$.
    \label{lem:PtEtEq}
\end{lemma}
\begin{proof}
    Any $\pB \in E_t$ can be written as
    \begin{equation}
        \pB = \frac{\One}{n} + \dl \dB,
        \label{eq:pBForm}
    \end{equation}
    where $0 < \dl \leq \dl_t:= \sqrt{t - 1/n}$ and $\dB$ is a direction vector
    coming from the set
    \begin{equation}
        F = \{ \dB \, | \, \dB \in \Rbb^{n}, \quad \dB.\dB \leq 1, \quad
        \text{and} \quad \One.\dB = 0\}.
        \label{def:fSet}
    \end{equation}
    One can show that any $\dB \in F$ satisfies
    \begin{equation}
        \dB_i \geq -\sqrt{\frac{n-1}{n}}.
        \label{eq:dlMin}
    \end{equation}
    To show the above inequality, we may solve
    \begin{equation}
        \underset{\dB \in F}{\text{minimize}} \; \dB_i
        \label{opt:prbl}
    \end{equation}
    where $1 \leq i \leq n$ is fixed but arbitrary. Notice the objective
    function above, $f(\dB) = \eB(i) \cdot \dB$ is linear in $\dB$ and the set $F$ is
    convex~(it has the form~\eqref{def:set}).  Inequalities
    defining~\eqref{def:fSet} are strict for the zero vector which lies in $F$.
    Thus, strong duality holds and the KKT
    conditions~\eqref{eq:kkt1}-\eqref{eq:kkt4}~(appropriately modified for this
    minimization problem~\eqref{opt:prbl}) are both necessary and sufficient.
    Since the minimum value $\dB_i < 0$ can always be decreased by scaling
    $\dB$ to have norm one, and must have $\dB^* \cdot \dB^* = 1$.
    Let $k = \sqrt{\frac{n-1}{n}}$, then
    \begin{equation}
        \dB^* = \frac{k}{n-1} (\One - \eB(i)) -k \eB(i), \quad
        \uB^*  = k, \quad 
        \text{and} \quad
        \wB^* = -\frac{1}{n},
    \end{equation}
    satisfy the KKT conditions. Thus, $f^{*} = f(\dB^*) = -k$, and this
    minimum value gives the lower bound in~\eqref{eq:dlMin}.  When $t \leq
    1/(n-1)$, inequality~\eqref{eq:dlMin} together with~\eqref{eq:pBForm} gives
    $\pB_i \geq 0$, i.e., $\pB \in P(t)$. 
\end{proof}

When $t> 1/(n-1)$, $P(t)$ is a strict subset of $E_t$, i.e., one can find $\pB
\in E_t$ that are not in $P(t)$. For instance, in~\eqref{eq:pBForm}, if we let
$\dl = \dl_t$ and $\dB = \sqrt{n/(n-1)}(\One/n - \eB_1)$, then we obtain a
vector in $E_t$ which is not in $P(t)$ since this vector has $\pB_1 < 0$.

\begin{lemma}
    Any $\pB \in E_t$ with $t > 1/n$ is an extreme point of $E_t$ if and only
    if $\pB \cdot \pB = t$. 
    \label{lem:extreme}
\end{lemma}
\begin{proof}
    If $\pB \cdot \pB < t$, then $\pB$ is not an extreme point since $\pB$ can always be
    re-written as a convex combination of $\One/n$ and a vector of the form
    in~\eqref{eq:pBForm} with $\dl = \dl_t$ and $\dB = (\One/n - \pB)/ ||
    (\One/n - \pB)||_2$.
    On the other hand, if $\pB \cdot \pB = t$, then it must be an extreme point
    of $E_t$. Assume $\pB \cdot \pB = t$ but $\pB$ is not an extreme point,
    then there must exist distinct $\fB_i \in E_t$ which satisfy
    \begin{equation}
        \pB = \sum_i a_i \fB_i 
        \label{eq:extremePt},
    \end{equation}
    where $\fB_i \neq \pB$, $0 < a_i < 1$ and $\sum_i a_i = 1$. The above
    equation, together with $\pB \cdot \pB = t$, implies
    \begin{equation}
        0 = \sum_i a_i (t - \pB \cdot \fB_i).
        \label{eq:zeroSum}
    \end{equation}
    The dot product,
    \begin{equation}
        \pB \cdot \fB_i \leq t,
        \label{eq:CSpq}
    \end{equation}
    can be shown using the Cauchy-Schwartz inequality~\eqref{eq:CS} with $\uB =
    \pB$ and $\vB = \fB_i$.
    The above equation implies that each term in the sum~\eqref{eq:zeroSum} is
    non-negative but the sum is zero. Thus each term must be zero and equality
    must hold in~\eqref{eq:CSpq}. This equality only holds~(see
    below~\eqref{eq:CS}) when $\pB = \lm_i \fB_i$ and $\lm_i \in \Rbb$.  Since
    $\pB, \fB_i \in E_t$, $\lm_i=1$. However, $\lm_i=1$ contradicts the
    assumption containing~\eqref{eq:extremePt}.

\end{proof}

\subsection{Special solutions}

Among the special instances of~\eqref{opt:maxFPhiPur-4} that can be solved
directly, the first is
when the maximum purity in~\eqref{eq:prbSet}, $t \geq 1$. As discussed
in Lemma~\eqref{eq:tGeq1Lemm}, when $t \geq 1$ $P(t) = P$.
The optimization~\eqref{opt:maxFPhiPur-4} reduces to
\begin{equation}
    \underset{\pB \in P}{\text{maximize}} \; \pB \cdot \qB 
    \label{opt:maxFPhiPur-5}
\end{equation}
Optimum and optimizer of~\eqref{opt:maxFPhiPur-5} are simply
\begin{equation}
    \pB^* = \eB(i), 
    \quad \text{and} \quad \pB^* \cdot \qB = \qB_i,
    \label{eq:pbStar0}
\end{equation}
respectively, where $i \in \text{argmax}_j \{ \qB_j \}$.
This equation follows from the fact that the maximum of a linear function on a
convex set lies at the extreme points of the set. The extreme points of $P$ are
mentioned below~\eqref{eq:prbSim}.

A second instance of~\eqref{opt:maxFPhiPur-4} with an algebraic solution is $t
\leq 1/n$. At $t < 1/n$, $P(t)$ is a null set~(see Lemma~\eqref{lem:EmptyFt})
and one may say there is no optimum or optimizer; when $t = 1/n$, $P(t)$ has a
single element $\One/n$~(see Lemma~\eqref{lem:SingleEt}), which becomes the
optimizer with optimum value $\sum_i \qB_i/n$.

A third case with an algebraic solution is $n=2$ where
\begin{equation}
    \pB^* = \frac{1}{2} \big( 
    \One + \sqrt{2t-1} (\eB_1 - \eB_2) \text{sign} (2\qB_1 - 1 ) 
    \big)
    \label{eq:n2Sol}
\end{equation}
can be proven to be an optimizer as follows. Since~\eqref{opt:maxFPhiPur-4} is the
maximum of a linear function over a convex set $P(t)$, its optimizer lies at an
extreme points of $P(t)$. At $n=2$ and $1/2 < t \leq 1$, there are two extreme
points of $P(t)$,
\begin{equation}
    \pB_{\pm} = \frac{1}{2} \big( \One \pm \sqrt{2t-1} (\eB_1 - \eB_2) \big).
\end{equation}
When $\qB_1 \geq 1/2$, $\pB_{+} \cdot \qB \geq \pB_{-} \cdot \qB$ and $\pB^* =
\pB_{+}$, otherwise $\pB_{-} \cdot \qB > \pB_{+} \cdot \qB$ and $\pB^* =
\pB_{-}$. 

In~\eqref{opt:maxFPhiPur-4} if $\qB = \One/n$ then the $\pB^* = \One/n$.
When $t \leq 1/(n-1)$ we get
\begin{equation}
    \pB^* = \bar{\pB} := 
    \frac{\One}{n} + \dl_t \mB,
    \label{eq:dirSol2}
\end{equation}
where $\dl_t = \sqrt{t - 1/n}$, $\mB:= (\qB - \One/n)/\kp$, and $\kp^2 :=
\qB \cdot \qB - 1/n$, i.e., $\mB$ is a unit vector along $\qB - \One/n$ with $\One
\cdot \mB = 0$. The first step to proving this claim is to observe when $1/n \leq
t \leq 1/(n-1)$, the set $P(t) = E_t$~(see Lemma~\ref{lem:PtEtEq}).
The second step is to show that~\eqref{eq:dirSol2} is the optimizer for
\begin{equation}
    \underset{\pB \in E_t}{\text{maximize}}  \; \pB \cdot \qB.
    \label{eq:prob2}
\end{equation}
The linear objective function in~\eqref{eq:prob2} takes its maximum value over
the convex set $E_t$ at an extreme point of this set.  Any $\pB \in E_t$ is
extreme if $\pB.\pB = t$~(see Lemma~\ref{lem:extreme}). This
extreme point can be written in the form~\eqref{eq:pBForm} where $\dl =
\dl_t$~(defined above~\eqref{eq:dlMin}).  A fixed vector $\qB$ with $\One \cdot
\qB = 1$~(a property we can always assume for $\qB$, see discussion
containing~\eqref{eq:qbPrime}) can be written as
\begin{equation}
    \qB = \frac{\One}{n} + \kp \mB,
\end{equation}
$\mB$ is defined below eq.~\eqref{eq:dirSol2}.
\begin{align}
    \pB \cdot \qB &= \frac{1}{n} + \dl_t \kp \dB \cdot \mB 
    \leq \frac{1}{n} + \dl_t \kp ||\dB||_2 \cdot ||\mB||_2
    = \frac{1}{n} + \dl_t \kp
\end{align}
where the first inequality is an application of the Cauchy-Schwartz
inequality~\eqref{eq:CS}, and the final equality uses the fact that $\dB$ and
$\mB$ are both unit vectors. When $\pB$ above is $\bar{\pB}$
in~\eqref{eq:dirSol2} the above inequality becomes an equality. This proves the
claim in~\eqref{eq:dirSol2}.

In general, the vector $\bar{\pB} \in E_t$~\eqref{eq:dirSol2} is the optimizer
of~\eqref{eq:prob2}. Since $P(t) \subseteq E_t$, the optimum of~\eqref{eq:prob2}
bounds the optimum value of~\eqref{opt:maxFPhiPur-4} from above. If $\bar{\pB}
\in P(t)$, i.e., $\bar{\pB} \geq 0$, then these optimum values of
~\eqref{opt:maxFPhiPur-4} and~\eqref{eq:prob2} become equal and $\bar{\pB}$ is
an optimizer of~\eqref{opt:maxFPhiPur-4}. 
One special case where $\bar{\pB} \geq 0$ is when
\begin{equation}
    \bar{\qB} := \frac{\One}{n} + \dl \kp \mB, \quad \text{with} 
    \quad \dl = \frac{1}{1-n \qB^{\downarrow}_n}, 
    \label{def:qBar}
\end{equation}
satisfies 
\begin{equation}
    \bar{\qB} \cdot \bar{\qB} \geq t.
    \label{eq:qBarNorm}
\end{equation}
One may prove this case as follows. Notice $\bar{\qB}$ is simply the vector
obtained by pushing $\qB$ to the boundary of the probability simplex along
the direction $\mB$, i.e.,
\begin{equation}
    \bar{\qB}^{\downarrow}_n = 0, \quad  \text{and} \quad \One.\bar{\qB} = 1,
    \label{eq:qBarProp}
\end{equation}
One can easily verify that
\begin{equation}
    \bar{\qB} \cdot \bar{\qB} = \dl^{2} \kp^2 + \frac{1}{n}.
    \label{eq:qBarDot}
\end{equation}
To show $\bar{\pB} \geq 0$ it is sufficient to show $\pB^{\downarrow}_n
\geq 0$.  Notice,
\begin{align}
    \bar{\pB}^{\downarrow}_n &= \frac{1}{n} + \dl_t \mB^{\downarrow}_n 
    = \frac{1}{n}(1 - \frac{\dl_t}{\dl \kp}) 
    \geq 0
\end{align}
where the first equality comes from the definition of $\bar{\pB}$
in~\eqref{eq:dirSol2}, the second comes from~\eqref{def:qBar}
and~\eqref{eq:qBarProp}. The final inequality makes use of the definition of
$\dl_t$ below~\eqref{eq:pBForm}, and equations~\eqref{eq:qBarNorm}
and~\eqref{eq:qBarDot}.

\subsection{Special objective functions}
\label{sec:basisVec}

We now provide some explicit solutions to the problem
in~\eqref{opt:maxFPhiPur-4}. In~\eqref{opt:maxFPhiPur-4} suppose $\qB =
\eB(i)$, where $1 \leq i \leq n$.  Using eq.~\eqref{eq:dirSol2}, one obtains
\begin{equation}
    \bar{\pB} = \frac{\One}{n} + \sqrt{\frac{nt - 1}{n-1} }(\eB_1 - \frac{\One}{n}),
\end{equation}
whose entries are non-negative for all allowed values $1/n \leq t \leq 1$. Thus
$\bar{\pB}$ above is the optimizer of~\eqref{opt:maxFPhiPur-4} with optimum value
\begin{equation}
    \bar{\pB} \cdot \eB(i) =  \frac{1}{n}(1 + \sqrt{(nt-1)(n-1)}).
\end{equation}

If $\Rbb^n = \Rbb^d \ot \Rbb^d$ and $\qB = \eB_1 \ot \One/d$,
then 
\begin{equation}
    \bar{\pB} = \frac{\One}{d^2} + \sqrt{\frac{td^2 - 1}{d-1}} 
    (\eB_1 \ot \frac{\One}{d} - \frac{\One}{d^2}),
    \label{eq:pBar2}
\end{equation}
has non-negative entries for $1/d^2 \leq t \leq 1/d$ and
\begin{equation}
    \bar{\pB} \cdot \qB = \frac{1}{d^2}(1 + \sqrt{(td^2 - 1)(d-1)})
    \label{eq:opt2}
\end{equation}
Thus for $\qB = (\eB_1 \ot \One)/n$ and $1/d^2 \leq t \leq
1/d$, the optimizer and optimal value to~\eqref{opt:maxFPhiPur-4}
are~\eqref{eq:pBar2} and~\eqref{eq:opt2}, respectively. 
For $1/d \leq t \leq 1$ the optimizer and optimal value are
\begin{equation}
    \pB^* = \eB_1 \ot \pB', \quad \text{and} \quad
    \pB^* \cdot \qB = \frac{1}{d},
    \label{eq:pbStar}
\end{equation}
respectively, where $\pB' \in \Rbb^d$ is any probability vector that satisfies
$\pB'.\pB' \leq t$. Proof for this claim is the following. Let
\begin{equation}
    \pB = \sum_{ij} \pB_{ij} \eB(i) \ot \eB_j, \quad
    \pB^1 := \sum_{ij} \pB_{ij} \eB(i), \quad \text{and} \quad
    \pB^2 := \sum_{ij} \pB_{ij} \eB_j,
\end{equation}
be the joint and marginal probability distributions, respectively. Using these
distribution, write
\begin{equation}
    \pB \cdot \qB = \frac{1}{d} \pB^1_{1} \leq \frac{1}{d};
\end{equation}
the above inequality becomes an equality when $\pB$ above is replaced with
$\pB^*$ in~\eqref{eq:pbStar}. Notice $\pB^*$ is a valid probability
distribution and $\pB' \cdot \pB' \leq t$ ensure that $\pB^* \cdot \pB^* \leq
t$.

\section{Algorithmic solutions}
\label{sec:algos}

\subsection{Recursive Algorithm}
\label{sec:recursiveAlgo}

In what follows we provide a recursive algorithm for solving the problem
in~\eqref{opt:maxFPhiPur-4},
\begin{equation}
    \underset{\pB \in P(t)}{\max} \; \pB \cdot \qB,
    \label{opt:glnVer}
\end{equation}
where we assume $\One. \qB = 1$ without loss of generality~(see discussion
containing~\eqref{eq:qbPrime}).

\RestyleAlgo{boxruled}
 \begin{algorithm}[H]
 \DontPrintSemicolon
 \SetAlgoLined
 \SetKwInOut{Input}{Input}\SetKwInOut{Output}{Output}
     \Input{$t \in \Rbb, \qB \in \Rbb^n$ with $\One \cdot \qB = 1$}
     \Output{$f(t, \qB, n) = (\pB^*, \pB^* \cdot \qB)$, $\pB^* \in 
     \underset{\pB \in P(t)}{\text{argmax}} \; \pB \cdot \qB$}
\BlankLine
 \uIf{$t \geq 1$}{
    $\pB^* \in \underset{\pB \in P}{\text{argmax}} \; \pB \cdot \qB$ \;
     \Return{($\pB^*, \pB^* \cdot \qB)$} }
\BlankLine
 \uElseIf{$t < 1/n$}{
     \Return{(None, $-\inf$)} }
\BlankLine
 \uElseIf{$t = 1/n$}{
     \Return{$(\One/n, 1/n)$}}
\BlankLine
 \uElseIf{$n = 2$}{
     $\pB^* = \frac{1}{2}\big(\One + \sqrt{2t-1} (1,-1) \text{sign}(2\qB_1 - 1) \big) $ \;
     \Return{($\pB^*, \pB^* \cdot \qB)$} }
\BlankLine
     $\bar{\pB} = \begin{cases}
         \frac{\One}{n} & \text{if} \; \qB = \One/n \\
         \frac{\One}{n} + \sqrt{t-\frac{1}{n}} \frac{\qB - \One/n}{||\qB - \One/n||} & \text{otherwise}
     \end{cases}
     $
\BlankLine
     \uElseIf{$t \leq \frac{1}{n-1}$ or $\bar \pB \geq 0$}{\Return{$(\bar{\pB}, \bar{\pB} \cdot \qB)$}}
 \Else{
     $(\bar{\qB}, i) = $ getClosestFace($\qB,n$)\;
     $\bar{\qB} = $ delete($\qB, i$)\;
     $(\aB^*, \aB^* \cdot \bar{\qB}) = f(t, \bar{\qB}, n-1)$ \;
     $\pB^* =$ append$(\aB^*, i, 0)$ \;
     \Return{$(\pB^*, \pB^*.\qB)$}
 }
 \caption{Maximum overlap with bounded purity}
     \label{alg:recursive}
\end{algorithm}

The key idea of this recursive algorithm is to solve certain base cases,
otherwise solve the problem in lower dimension. The simplest base cases are
those where the parameter $t \geq 1$, or $t \leq 1/n$, or $n=2$, or $\qB =
\One/n$.  A more complicated base case occurs when $\qB$ or $t$ is such that
solving~\eqref{opt:maxFPhiPur-4} is equivalent to solving the problem without
the $\pB \geq 0$ constraint. In all other cases, it suffices to solve the
problem by restricting $\pB$ to be on an $(n-1)$-dimensional boundary of the
probability $n$-simplex $P$.
This final step along with all other steps of the algorithm are listed as a
pseudo-code labeled Algorithm~\ref{alg:recursive}. In what follows, we 
describe various steps of this algorithm.

\begin{itemize}
    \item In line 1, if the base case $t \geq 1$ is found to be true, then its
        solution, see discussion containing~\eqref{eq:pbStar0}, is returned.
        Similar comments apply for lines 4 and 6.
    \item In Line 8, if $n=2$ then its solution~(see~\eqref{eq:n2Sol}) is
        returned.
    \item Line 11 defines $\bar{\pB}$. As discussed in the paragraphs
        containing~\eqref{eq:prob2} and~\eqref{def:qBar}, this vector solves
        the problem at hand~\eqref{opt:glnVer} when $t \leq 1/(n-1)$ or
        $\bar{\pB} \geq 0$ or $\qB = \One/n$ or $\bar{\qB} \cdot \bar{\qB} \geq
        t$~(see discussion containing~\eqref{def:qBar}). The vector is returned
        appropriately in lines 12 and 13.
    \item In line 14 we begin a recursion step.  First, in line 15 the
        subroutine {\tt getClosestFace} gets $\bar{\qB}$, defined
        in~\eqref{def:qBar}, along with the index $i$ where $\qB_i \leq \qB_j$
        for all $1 \leq j \leq n$. Next, in line 16, the $i^{\text{th}}$ entry
        of $\bar{\qB}$ is popped out, reducing the dimension of $\bar{\qB}$
        from $n$ to $n-1$. With this $(n-1)$-dimensional vector as input, we
        recurs to solve the problem in $n-1$ dimensions. The solution is
        obtained in line 17. This solution, an $(n-1)$-dimensional vector
        $\aB^{*}$, is turned into an $n$-dimensional vector $\pB^*$ by
        re-inserting $0$ at the $i^{\text{th}}$ coordinate in line 18. 
\end{itemize}
To prove this recursion procedure is correct, it suffices to show that if
\begin{equation}
    \pB^* \in \underset{\pB \in P(t)}{\text{argmax}} \; \pB. \qB 
    \quad \text{and} \quad 
    \bar{\qB} \cdot \bar{\qB} < t \quad \text{then} \quad
    \pB^*_i = 0.
    \label{eq:claim}
\end{equation}
Proving the above claim is sufficient because showing $\pB^*_i = 0$ implies
that the $i^{\text{th}}$ coordinate can be dropped and the problem can be
solved for the rest of the coordinates. The additional condition, $\bar{\qB}
\cdot \bar{\qB} < t$, can always be assumed since whenever $\bar{\qB} \cdot
\bar{\qB} \geq t$ the optimizer $\bar{\pB}$~(see discussion containing
eq.~\eqref{def:qBar}) is returned in previous lines 12 and 13.

To prove~\eqref{eq:claim} use the definitions of $P(t)$ in~\eqref{eq:prbSet},
$\bar{\qB}$ in~\eqref{def:qBar}, and the fact that in~\eqref{def:qBar} $\dl
\geq 0$ to obtain
\begin{equation}
    \pB^{*} \in \underset{\pB \in P(t)}{\text{argmax}} \; \pB. \qB  
    \iff
    \pB^{*} \in \underset{\pB \in P(t)}{\text{argmax}} \; \pB. \bar{\qB}.  
    \label{eq:optProb}
\end{equation}
Any $\pB^*$ above must satisfy $\pB^*_i = 0$ for some $i$ where we are given
$1/n < t < 1$, $\bar{\qB} \cdot \bar{\qB} < t$, $\qB_i = 0$, and $\qB \neq 0$.

Notice $\pB. \bar{\qB}$ is linear, $P(t)$ is a convex set, and $\One/n \in
P(t)$ satisfies $\One/n > 0$, $\One. \One/n = 1$, and $(\One/n) \cdot (\One/n)
< t$.  As a result, strong duality holds for the optimizations on the right
side of ~\eqref{eq:optProb} and the KKT
conditions~\eqref{eq:kkt1}-\eqref{eq:kkt4} are both necessary and sufficient
for optimality.
Assume $\pB^{*}$ is optimal but $\pB^{*}_i \neq 0$ for all $i$. The KKT
conditions imply that $\pB^{*} \in P(t)$ satisfies
\begin{equation}
    \bar{\qB}_j + \uB^*_j - 2w^* \pB_j^* - z = 0 \quad \forall j
    \label{eqA:KKT1}
\end{equation}
while $\uB^* \geq 0$ satisfies
\begin{equation}
    \uB_j^* \pB_j^* = 0 \quad \forall \; 1 \leq j \leq n, 
    \label{eqA:KKT2}
\end{equation}
and $w^* \geq 0$ satisfies
\begin{equation}
    w^* (\pB^* \cdot \pB^* - t) = 0.
    \label{eqA:KKT3}
\end{equation}
Since $\pB^{*}_i \neq 0$ and $\pB^* \in P(t)$ we get $\pB^{*} > 0$. This final
inequality, together with~\eqref{eqA:KKT2} gives $\uB^* = 0$. If we assume $w^* =
0$, then~\eqref{eqA:KKT1} gives $\bar{\qB}_j = z$ for all $j$. This is
inconsistent, $\bar{\qB}^{\downarrow}_n = 0$ gives $z=0$ which implies $\bar{\qB}=0$ in
contradiction to $\sum_i \bar{\qB_i} = 1$. Thus $w^* \neq 0$, as a result
from~\eqref{eqA:KKT3}, $\pB^* \cdot \pB^* - t = 0$. Setting $\uB = 0$
in~\eqref{eqA:KKT1} and $\One \cdot \pB^* = \One \cdot \bar{\qB} = 1$ gives
\begin{equation}
    \pB_j^* = \begin{cases}
        \frac{1}{n} + \frac{1}{2w^*}(\bar{\qB_j} - \frac{1}{n}) & \text{if} 
        \quad j \neq i \\
        \frac{1}{n}(1 -  \frac{1}{2w^*}) & \text{otherwise}
    \end{cases}.
\end{equation}
Since $\pB^* > 0$ we get 
\begin{equation}
    w^* > 1/2
    \label{eq:vGreater}
\end{equation}
The above equation, together with $\pB^* \cdot \pB^* = t$, 
$\One \cdot \pB^* = \One \cdot \bar{\qB} = 1$ gives
\begin{equation}
    t - \frac{1}{n} = \frac{1}{4(w^*)^2}( \bar{\qB} \cdot \bar{\qB} -\frac{1}{n}).
\end{equation}
Since $\bar{\qB} \cdot \bar{\qB} < t$ we get $w^* < 1/2$. Which
contradicts~\eqref{eq:vGreater}. Thus, we cannot assume $\pB^*_i \neq 0$ for
all $1 \leq i \leq n$. Suppose $\pB^*_j =0$ for some $j \neq i$ where
$\bar{\qB}_i = 0$. Then consider another vector $\tilde{\pB}$ obtained by
swapping the $i^{\text{th}}$ and $j^{\text{th}}$ coordinate of $\pB^*$.  Since
$\bar{\qB} \geq 0$~(see eq.~\eqref{def:qBar}), $\pB^* \geq 0$, and $\tilde{\pB}
\geq 0$ we get $\pB^* \cdot \qB \leq \tilde{\pB} \cdot \qB$.  As a result, 
we can always assume $\pB^*_i = 0$ when $\bar{\qB}_i = 0$.

\subsection{Dual Algorithm}
\label{sec:dualAlgo}

The optimization problem in~\eqref{opt:glnVer} can also be solved using
Lagrange duality~(see Sec.~\ref{sec:opt}). The dual formulation is not needed
when (1) $t \geq 1$, this problem is easily solved~(see discussion
containing \eqref{eq:pbStar0}), (2) $t < 1/n$, here~\eqref{opt:glnVer} is
infeasible, and (3) when $t = 1/n$, here~\eqref{opt:glnVer} has an algebraic
solution. In all other cases, we construct the Lagrangian
for~\eqref{opt:glnVer},
\begin{align}
    \label{eq:lagrange}
    L(\pB, \uB, w, z) = -w \pB \cdot \pB + \pB \cdot (\uB + \qB -z \One) + (wr +z),
\end{align}
where $\uB \geq 0$, $w \geq 0$, and $z \in \Rbb$ are called Lagrange
multipliers or dual variables.  The Lagrange dual function~\eqref{def:dualObj}
takes the form
\begin{equation}
    g(\uB, w, z) = 
    \begin{cases}
        \infty & \text{if} \; w = 0 \; \text{and} \; \uB \neq z\One - \qB \\
        wt + z & \text{if} \; \uB = z\One - \qB \\
        \frac{1}{4w} || \uB + \qB - z\One||_2^2 + wt + z & \text{if} \; w \neq 0
        \; \text{and} \; \uB \neq z\One - \qB 
    \end{cases}.
    \label{eq:gCases}
\end{equation}
Since $1/n < t$, there is a strictly feasible probability vector, $\pB = \One/n$
satisfying $\pB > 0$, $\One \cdot \pB = 1$, and $\pB \cdot \pB  = 1/n < t$. As
stated in the paragraph containing~\eqref{eq:strictFeas}, such strict
feasibility implies the optimal value of the dual,
\begin{equation}
    g^* = \pB^* \cdot \qB.
    \label{eq:strongDuality}
\end{equation}
When $1/n < t < 1$, one may compute $g^*$. For this computation, define 
\begin{equation}
    \vB_i(z) := \max(0, \qB_i - z),
    \label{eq:vDef}
\end{equation}
and
\begin{equation}
    h(z) := 
    \begin{cases} 
        \sqrt{t}||\vB(z)||_2 + z & \text{if} \; z < \qB^{\downarrow}_1 \\
        z & \text{otherwise} 
    \end{cases}.
    \label{eq:iDef}
\end{equation}
One can show that 
\begin{equation}
    h(z) \leq g(\uB, w, z)
    \label{eq:iIneqG}
\end{equation}
for all $\uB \geq 0$, $w \geq 0$ and $z \in \Rbb$ and equality holds for
special values of $\uB$ and $w$. The inequality is obvious when $g(\uB, w,
z)$ in \eqref{eq:gCases} is $\infty$. Otherwise, consider two
cases. First, $z \geq \qB^{\downarrow}_1$, here $h(z) = z$ is
smaller than $g(\uB, w, z)$ in~\eqref{eq:gCases} and equality holds when $\uB =
z\One - \qB \geq 0$ and $w = 0$. Second, $z < \qB^{\downarrow}_1$, in this
case equality $\uB = z \One - \qB$ in~\eqref{eq:gCases} cannot arise for
any $\uB \geq 0$ since
$\uB_i < 0$ whenever $\qB_i = \qB^{\downarrow}_1$. The only relevant case
from~\eqref{eq:gCases} is $w \neq 0$ and $\uB \neq z \One - \qB$, i.e.,
\begin{align}
\begin{aligned}
    g(\uB, w, z) &= \frac{1}{4w} || \uB + \qB - z\One||_2^2 + wt + z\\
    &\geq \frac{1}{4w} || \uB^* + \qB - z\One||_2^2 + wt + z \\
    &= \frac{1}{4w} ||\vB(z)||_2^2 + (wt + z) \\
    &\geq \sqrt{t} ||\vB(z)||_2  + z \\
    &= h(z). 
\end{aligned}
    \label{eq:gBounds}
\end{align}
where the first inequality follows from choosing 
\begin{equation}
    \uB^*_i = \max(0, z - \qB_i),
    \label{eq:uBVal}
\end{equation}
the value of $\uB \geq 0$ which minimizes $g(\uB, w,z)$ in the present case;
the second equality in~\eqref{eq:gBounds} follows from the definition of $\vB(z)$
in~\eqref{eq:vDef}, the second inequality in~\eqref{eq:gBounds} follows by
choosing $w$ to be
\begin{equation}
    w^* = \frac{1}{\sqrt{4t}}||\vB(z)||_2,
    \label{eq:wBVal}
\end{equation}
the value of $w > 0$ that minimizes expression on the right side of the second
equality in~\eqref{eq:gBounds}, and the final equality in~\eqref{eq:gBounds}
follows from the definition of $h(z)$ in~\eqref{eq:iDef}. It follows that the
inequalities in~\eqref{eq:gBounds} become equalities when $\uB = \uB^*$, $w =
w^*$.

Notice, $h(z)$ is convex in $z \in \Rbb$ since $h(z)$ can be obtained from the convex
function $g(\uB, w, z)$~(see discussion below~\eqref{eq:gCases}) at special
values of $\uB$ and $w$. 
To minimize $h(z)$ one need not search over all of $\Rbb$ but some interval
$[z_{\min}, z_{\max}]$ of $\Rbb$.
In particular one can show that for some $z_{\min}$ and $z_{\max}$,
$h(z)$ is a decreasing function for $z < z_{\min}$ and an increasing function
for $z> z_{\max}$. One may find $z_{min}$ and $z_{max}$ by computing
\begin{equation}
    h'(z):= \frac{d h(z)}{dz} = \begin{cases}
        1 - \sqrt{t} \frac{|| \vB(z) ||_1}{|| \vB(z) ||_2} & \text{if} \; z < \qB^{\downarrow}_1 \\
        \text{None} & \text{if} \; z =\qB^{\downarrow}_1 \\
        1 & \text{otherwise}
    \end{cases}.
    \label{eq:iDer}
\end{equation}
Since $h'(z) > 0$ for $z > \qB^{\downarrow}_1$, let
\begin{equation}
    z_{\max} = \qB^{\downarrow}_1.
    \label{eq:zMax}
\end{equation}
One can choose 
\begin{equation}
    z_{\min} = \qB^{\downarrow}_n - 2|s|
    \label{eq:zMin}
\end{equation}
where
\begin{equation}
    \frac{1 - \sqrt{t}}{\sqrt{n}(\sqrt{nt} - 1)} ||\qB - \qB^{\downarrow}_n \One ||_1 < s.
    \label{eq:kIneq}
\end{equation}
The choice can be justified by showing $h'(z) < 0$ when $z < z_{\min}$.
Notice, at $z = \qB^{\downarrow}_n - |s|$,
\begin{equation}
    ||\vB(z)||_1 = || \qB - \qB^{\downarrow}_n \One + |s| \One||_1 = 
 || \qB - \qB^{\downarrow}_n||_1 + |s|n 
    \label{eq:vOne}
\end{equation}
where the first equality follows from definition of $\vB(z)$ in~\eqref{eq:vDef}
and value of $z$ chosen above, the second equality uses $\qB -
\qB^{\downarrow}_d \One \geq 0$. Next,
\begin{equation}
    ||\vB(z)||_2 = || \qB - \qB^{\downarrow}_n \One + |s| \One||_2 
    \leq || \qB - \qB^{\downarrow}_n \One ||_2 + |s|\sqrt{n}
    \leq || \qB - \qB^{\downarrow}_n \One ||_1 + |s|\sqrt{n}
    \label{eq:vTwo}
\end{equation}
where the first inequality follows from a triangle inequality, $||\aB + \bB||
\leq ||\aB|| + ||\bB||$, and the second from~\eqref{eq:vecNormIneq}.
Notice, inequality~\eqref{eq:kIneq} can be rearranged as follows,
\begin{equation}
    || \qB - \qB^{\downarrow}_n \One ||_1 + s \sqrt{n} <
        \sqrt{t} (|| \qB - \qB^{\downarrow}_n \One ||_1 + sn ).
\end{equation}
Using~\eqref{eq:vOne} and~\eqref{eq:vTwo} in the above equation, we find
\begin{equation}
    ||\vB||_2 < \sqrt{t} ||\vB||_1.
\end{equation}
The above equation, together with the definition of $h'(z)$~\eqref{eq:iDer}
implies $h'(z)<0$ at $z = \qB^{\downarrow}_n - |s|$.

Let $1/n < t < 1$, $z^* \in \text{argmin} \; h(z)$, and $h^* = h(z^*)$
then~\eqref{eq:iIneqG} implies
\begin{equation}
    h^* = g^*.
\end{equation}
Notice $g^*$ obtained this way gives the optimum value of~\eqref{opt:glnVer}.
One can also obtain the optimizer $\pB^*$ of~\eqref{opt:glnVer} using
$z^*$ and KKT conditions~(see Sec.~\ref{sec:opt}).
Notice, $z_{\min} \leq z^* \leq z_{\max}$. If $z^* = z_{\max} =
\qB^{\downarrow}_1$, then
\begin{equation}
    h^* = \qB^{\downarrow}_1  = g(z_{\max} \One - \qB, 0, z_{\max}),
\end{equation}
here $\uB = z_{\max} \One - \qB$, $w = 0$, and $z = z_{\max}$ satisfy 
dual feasibility~\eqref{eq:kkt4}.
Let $S = \{ i \; | \qB_i = \qB^{\downarrow}_1 \}$ then
\begin{equation}
    \pB^*_i = \begin{cases}
        1/|S| & \text{if} \; \qB_i = \qB^{\downarrow}_1  \\
        0 & \text{otherwise}
    \end{cases}.
    \label{eq:pBStar0}
\end{equation}
Here, $\pB^* \geq 0$ and $\sum_i \pB^*_i = 1$, in addition one may show 
that $\pB^* \cdot \pB^* = 1/|S|  \leq t$. This inequality can be shown as follows,
let $z = \qB^{\downarrow}_1 - \ep$, $\ep > 0$ s.t.
$\qB^{\downarrow}_{|S| + 1} < z < \qB^{\downarrow}_1$~(such an $\ep$ always exists
unless $\qB = \One/n$, in which case $|S| = n$ and thus $1/n \leq t$ by assumption) then
by definition
\begin{equation}
    h^* = \qB^{\downarrow}_1 < h(\qB^{\downarrow}_1 - \ep) = |\ep|(\sqrt{|S|t}
    - 1) + \qB^{\downarrow}_1
\end{equation}
The above inequality implies $0 \leq \sqrt{t|S|} - 1$, which in turn implies
$\frac{1}{|S|} \leq t$. It follows that $\pB^*$ satisfies primal
feasibility~\eqref{eq:kkt3}.
Complementary slackness~\eqref{eq:kkt2}, easily follows since $\pB^*_i \uB^*_i
= \pB_i(\qB^{\downarrow}_{1} \One - \qB_i) = 0$ and $w(t -\pB^* \cdot \pB^*)
= 0$. Finally, stationarity~\eqref{eq:kkt1} $\nabla_{\pB} L(\pB, \uB, w, z) =
0$ at $\pB =  \pB^*$, $w = 0$, and $\uB = \qB^{\downarrow}_{1} \One - \qB$, and
$z^* = \qB^{\downarrow}_1$ can be checked.

Alternatively, if $z^* < z_{\max}$ then 
\begin{equation}
    \pB^* =  \frac{\sqrt{t}}{||\vB(z^*)||_2} \vB(z^*),
    \label{eq:pBStar}
\end{equation}
where $\vB(z)$, defined in~\eqref{eq:vDef}, is a non-zero vector, $\uB =
\uB^*(z^*)$ in~\eqref{eq:uBVal}, $w = w^*(z^*) > 0$ in~\eqref{eq:wBVal} satisfy
the KKT conditions. Dual feasibility~\eqref{eq:kkt4} is easy to check, primal
feasibility~\eqref{eq:kkt3} can be checked using the fact that $h'(z^*) =
0$~\eqref{eq:iDer}, complementary slackness~\eqref{eq:kkt2} and
stationarity~\eqref{eq:kkt1} can be checked directly from the form
of~\eqref{eq:pBStar}.

\begin{lemma}
    The optimizer $\pB^*$ can always be chosen to have purity $t$.
    \label{lem:purityEq}
\end{lemma}
We have two different types of optimizers, $\pB^*$ in~\eqref{eq:pBStar0} and
$\pB^*$ in~\eqref{eq:pBStar}. This second optimizer in~\eqref{eq:pBStar} has
purity exactly $t$. The first optimizer
$\pB^*$ in~\eqref{eq:pBStar0} satisfies $\pB^* \cdot \pB^* \leq t$, if the
inequality happens to be an equality the purity of this optimizer is exactly
$t$. Otherwise, let $\pB^*(\dl) = (1-\dl) \pB^* + \dl \eB(i)$ where $i \in S$
is defined above eq.~\eqref{eq:pBStar0}, and
$0 \leq \dl \leq 1$ be a one parameter family of optimizers. These satisfy
$\pB^*(\dl) \cdot \qB = \pB^* \cdot \qB$.  Notice $a(\dl) := \pB^*(\dl) \cdot
\pB^*(\dl) - t$ is a smooth function of $\dl$, $a(0) < 0$, and $a(1) \geq 0$,
thus from the intermediate value theorem, there exists a value $0 < \dl' < 1$
where $a(\dl') = 0$, i.e. $ \pB^*(\dl') \cdot \pB^*(\dl') = t$.  Since $a(\dl)$
is convex in $\dl$, with $a(0) < 0$ and $a(1) \geq 0$ one can perform binary
search over $\dl$ to find $\dl'$ such that $a(\dl') = 0$.

Using insights from solving the Lagrange dual formulation, one may numerically
optimize~\eqref{opt:maxFPhiPur} using Algorithm~\eqref{alg:main2}. In this
algorithm, the first six lines solve obvious bases cases $t \geq 1$, and $t
\leq 1/n$. In lines 9-14, one optimizes the dual function by minimizing $h(z)$,
a convex function, over a bounded interval~(see comments below
eq.~\eqref{eq:minFour} and~\cite{PressTeukolskyEA92, VirtanenGommersEA20}) to
obtain the optimum value of the dual. In the next lines 16-21, the optimizer
for the primal problem is constructed. 

\RestyleAlgo{boxruled}
 \begin{algorithm}[H]
 \DontPrintSemicolon
 \SetAlgoLined
 \SetKwInOut{Input}{Input}\SetKwInOut{Output}{Output}
     \Input{$t \in \Rbb, \qB \in \Rbb^n$ with $\One.\qB = 1$}
     \Output{$h(t, \qB, n) = (\pB^*, \pB^* \cdot \qB )$, $\pB^* \in 
     \underset{\pB \in \FC_t}{\text{argmax}} \; \pB.\qB$}
 \uIf{$t \geq 1$}{
    $\pB^* \in \underset{\pB \in \PC}{\text{argmax}} \; \pB.\qB$ \;
     \Return{($\pB^*, \pB^*.\qB)$} }
\BlankLine
 \uElseIf{$t < 1/n$}{
     \Return{(None, $-\inf$)} }
\BlankLine
 \uElseIf{$t = 1/n$}{
     \Return{$(\One/n, 1/n)$}}
\BlankLine
 \Else{
     $k = (1 - \sqrt{t}) || \qB  - \qB^{\downarrow}_n \One||_1 / (\sqrt{n}(\sqrt{nt} - 1))$ \;
     $z_{\min} = \qB^{\downarrow}_n - 2|s|$ \;
     $z_{\max} = \qB^{\downarrow}_1$ \;
     $\vB_h(z) = \max(0, \qB_i - z)$ \;
     $h(z) = \begin{cases} 
         \sqrt{t} || \vB(z) ||_2 + z & \text{if} \; z < \qB^{\downarrow}_1 \\
         z & \text{otherwise}
     \end{cases}.$ \;
    $z^* \in \underset{z_{\min} \leq z \leq z_{\max}}{\text{argmin}} \; h(z)$
    }
\BlankLine
     \If{$\vB(z^*) = 0$}{
     $S = \{i \; | \qB_i = \qB^{\downarrow}_1 \}$ \;
     $\pB^*_i = \begin{cases} 
         1/|S| & \text{if} \; \qB_i = \qB^{\downarrow}_1 \\
         0 & \text{otherwise}
     \end{cases}$
     }
\BlankLine
     \Else{
    $\pB^* = \sqrt{r} \vB/||\vB||_2$
     }
\BlankLine
     \Return{$( \pB^*, h(z^*) )$}
 \caption{Maximum overlap with bounded purity}
     \label{alg:main2}
\end{algorithm}

\subsection{Numerical results}
\label{sec:numerics}

We compare the performance of algorithms~\eqref{alg:recursive}
and~\eqref{alg:main2} against the CVXOPT~\cite{SraNowozinEA11, AndersenDahlEA21}
solver accessed via the PICOS~\cite{SagnolStahlberg22} interface.
We carry out this comparison for various different dimensions $n$. For each
dimension, we randomly sample 50 values of $\qB$ to define the objective
function and plot the average time to solution.
In Fig.~\ref{fig:timeToSolFit1} we plot the average time to solution taken by
the dual algorithm. The plot best fits a linear curve. This linear curve is
expected since the algorithm takes $O(n)$ time, a fixed time to minimize $h(z)$
and $O(n)$ operations to read $\qB$ and reconstruct the optimizer.
In Fig.~\ref{fig:timeToSolFit2} we plot the average time to solution taken by
the recursive algorithm. These points for the recursive algorithm best fits a
quadratic curve. This is expected from the nature of the recursive algorithm,
in each step of the recursion the algorithm taken $O(n)$ steps, and there are
at most $O(n)$ recursive calls, thus taking a total of $O(n^2)$ time.
Finally data points from CVXOPT plotted in Fig.~\ref{fig:timeToSolFit3} best
fits a cubic curve.
The average time to solution is plotted for all three algorithms jointly on a
logarithmic scale in Fig.~\ref{fig:timeToSol}.
It indicates that the average time to solution for CVXOPT is longer than that
for the recursive algorithm, which in turn takes longer than the recursive
algorithm. For large dimension, $n \simeq 8000$, average time to solution for
CVXOPT is roughly two to three orders of magnitude more than the recursive
algorithm and this recursive algorithm takes one to two order of magnitude more
time than the dual algorithm on average.

\begin{figure}[h]
	\centering
    \includegraphics[width=0.8\textwidth]{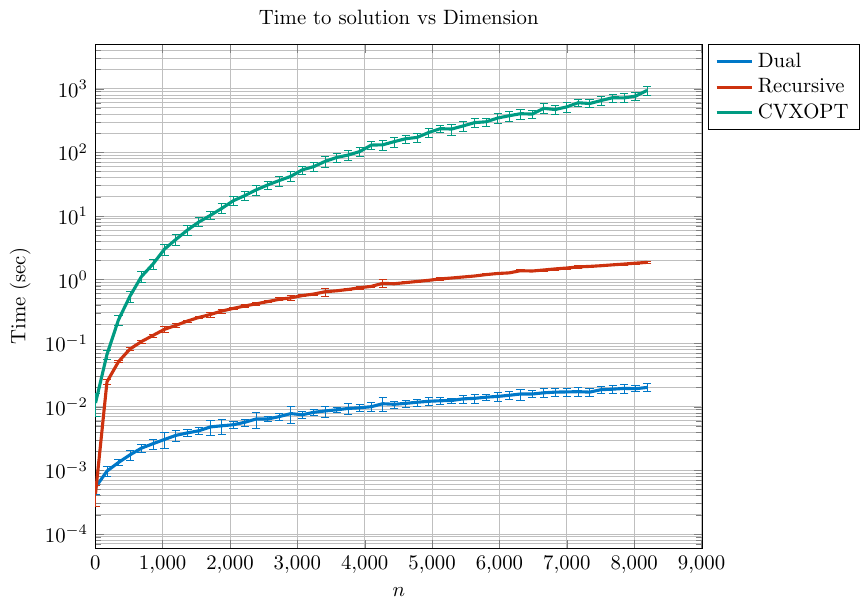}
    \caption{Plot of average time to solution for optimization
    problem~\eqref{opt:glnVer} using various algorithms against the problem
    dimension $n$. For each $n$ the average is taken over 50 random problem
    instances. Each instance is solved using the recursive
    algorithm~\eqref{alg:recursive}, the dual algorithm~\eqref{alg:main2}, and
    CVXOPT.}
    \label{fig:timeToSol}
\end{figure}

\begin{figure}[h]
	\centering
    \includegraphics[width=0.8\textwidth]{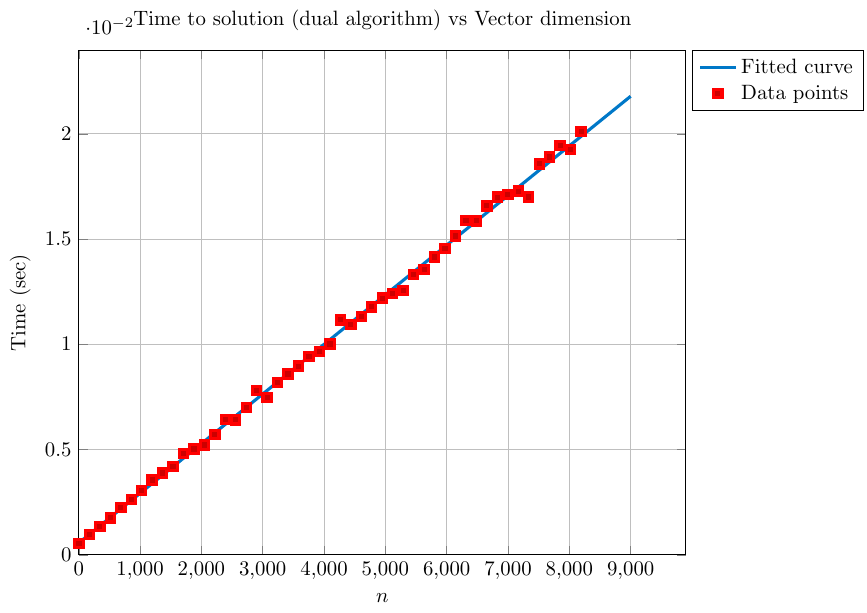}
    \caption{A linear curve fitting the data points for time to solution using
    the dual algorithm~\eqref{alg:main2}}
    \label{fig:timeToSolFit1}
\end{figure}

\begin{figure}[h]
	\centering
    \includegraphics[width=0.8\textwidth]{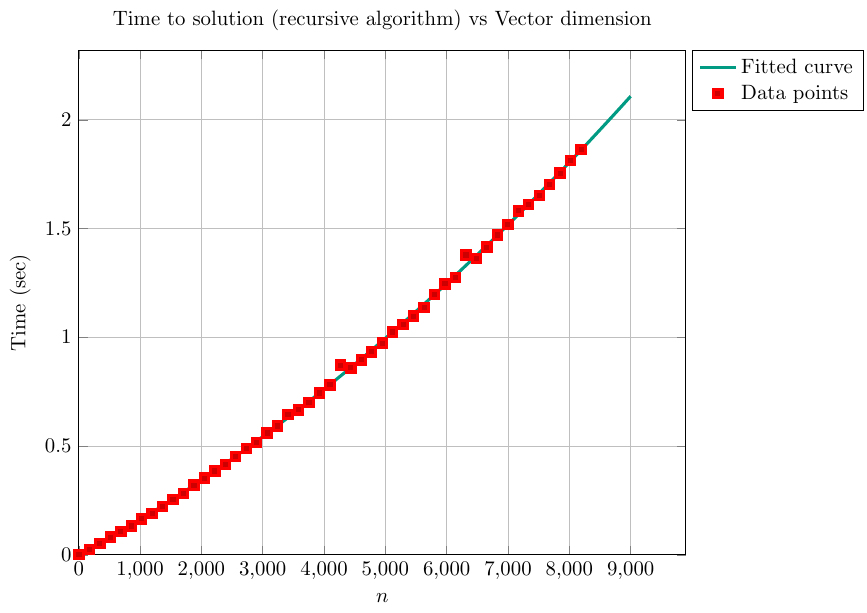}
    \caption{A quadratic curve fitting the data points for time to solution
    using recursive algorithm~\eqref{alg:recursive}.}
    \label{fig:timeToSolFit2}
\end{figure}

\begin{figure}[h]
	\centering
    \includegraphics[width=0.8\textwidth]{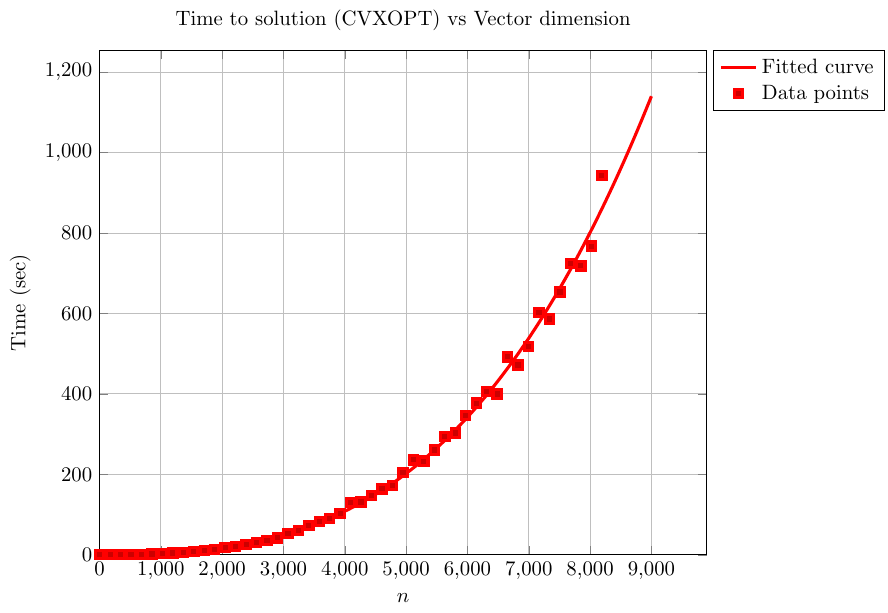}
    \caption{A cubic curve fitting the data points for time to solution using
    CVXOPT.}
    \label{fig:timeToSolFit3}
\end{figure}

\clearpage

\section{Discussion}
\label{sec:discuss}

Evolution of a state due to noise is described using a quantum channel,
equivalently evolution of observables is described using a dual map.
In the presence of such noise we find, in-terms of the dual map, exact
expressions for the maximum fidelity of preparing a target pure state 
and the minimum energy a system can be prepared in~(see Sec.~\ref{sec:noisyVal}).
These results can be proven succinctly~(using an optimization and matrix
analysis lens) and are applicable broadly.
One may rightly expect noise to increase and decrease a system's energy over
time.  We find that no decrease can be observed for the ground-state energy as
the least eigenvalue of observables increases under noisy evolution described by
a quantum channel.  This
mathematical result supports the physical intuition that noise can't help cool
a system's ground-state. It may be interesting to use this result to constrain
other heuristic approaches to study noisy evolution of Hamiltonians.

Description of a noisy quantum channel grows exponentially with system size and
may not be readily available even for smaller systems of interest.  Thus, it is
important to asses the effect of noise without its complete
description. 
To this end, in Sec.~\ref{sec:purLim} we propose including a measurable effect
of noise in one's assessment instead. This effect is a decrease in purity.
Constraining the system to have some maximum possible purity, we show how to
estimate a system's ability to perform some fundamental task. There are
cases where the constraint gives an upper bound on the performance~(see
discussion containing eq.~\eqref{eq:renyConst}) of fundamental tasks under
noise.

The tasks of interest studied in this work include preparing a pure state with
high fidelity, measuring the ground-state energy of a system, and sharing high
fidelity entanglement, each carried out using states with some maximum purity.
In each case the respective metric of performing the task can be found by
maximizing the expectation value of some $d \times d$ observable over states of
bounded purity. 
We map this problem~(at a cost of at most $O(d^3)$ steps incurred when the
eigendecomposition of the observable is not known) from the space of matrices
to those over probability vectors in such a way that the optimization remains
convex~(see Sec.~\ref{sec:reformprob}). This optimization can be solved
using state-of-the-art numerical solvers such as CVXOPT. We improve upon CVXOPT
by proposing two numerical algorithms~(see Sec.~\ref{sec:algos}). While these
algorithms perform the same optimization, they do so in two very different
ways. We include both since one may be easier to extend than the other and they
each help benchmark the other.

When optimizing over $n$-dimensional probability vectors, our first recursive
algorithm~(see Alg.~\ref{alg:recursive}), takes $O(n^2)$ steps, while our
second algorithm~(see Alg.~\ref{alg:main2}) makes use of duality theory in
optimization and takes only $O(n)$ steps. Numerical experiments indicate that
both our algorithms take less time than CVXOPT, which appears to take $O(n^3)$
time. Our numerics also indicate that for larger dimensions, the advantages of
using our algorithms can be substantial, for instance the dual algorithm appears
roughly four orders of magnitude faster than CVXOPT~(see discussion
in Sec.~\ref{sec:numerics}) in our numerics.

Another task of interest for modest system sizes is quantum tomography. Here
again, noise in state preparation and measurement affects the estimate. There
are ways to modify standard estimators to incorporate noise. However, the
modifications can come from statistical assumptions, require hyper-parameter
tuning, or can require knowledge of the full noise channel
itself~\cite{BlumeKohout10, TeoZhuEA11, Siddhu19}. By contrast, our work allows
one to build on a standard estimator, the maximum likelihood estimator~(MLE)
from~\cite{SmolinGambettaEA12a}, while incorporating the effect of noise by
including a non-linear constraint, the system's purity which is measurable~(see
Sec.~\ref{sec:tomo}).
Purity constraints can result in both convex and non-convex optimization
problems. Even though general non-convex optimization cannot be solved
optimally, we show how to use our algorithms to solve both convex and
non-convex versions of these tomographic reconstruction problems.

To solve the convex version of the problem, we also propose a third algorithm
in Sec.~\ref{sec:minD}.  From an optimization theory point of view, this
algorithm projects, under Frobenius norm, onto the set of unit-trace positive
semi-definite operators with some fixed maximum Frobenius norm~(see discussion
below~\eqref{eq:minFour}).  The projection requires a one-parameter
optimization which is straightforward. 

\section{Acknowledgements}
We thank Sergey Bravyi and Charles Bennett for helpful discussions and
Mohammad Alhejji for useful comments on the manuscript.

\printbibliography[heading=bibintoc]

\end{document}